%% file: main.tex
\documentclass[12pt]{article}
\usepackage{jheppub}

\makeatletter
\def\@fpheader{\relax}
\makeatother

\usepackage[skip=0.4\baselineskip plus 1pt minus 1pt, indent=1.5em]{parskip}
\usepackage{graphicx, amsmath, amsthm, amssymb, mathtools, booktabs, multirow, multicol, csquotes}
\usepackage{tcolorbox}
\usepackage{hyperref}
\hypersetup{
    colorlinks=true,
    allcolors=blue
    }
\usepackage[capitalize]{cleveref}

\RequirePackage{tikz}
\usetikzlibrary{positioning, arrows.meta, calc, tikzmark}

\newtheorem{theorem}{Theorem}
\newtheorem{prop}{Proposition}

\newtheorem{conj}{Conjecture}
\newtheorem{coro}{Corollary}

\crefname{conj}{Conjecture}{Conjectures}

\DeclarePairedDelimiter{\braces}{\lbrace}{\rbrace}
\DeclarePairedDelimiter{\bracks}{\lbrack}{\rbrack}
\DeclarePairedDelimiter{\parens}{\lparen}{\rparen}
\DeclarePairedDelimiter{\floor}{\lfloor}{\rfloor}
\DeclarePairedDelimiter{\abs}{\lvert}{\rvert}
\DeclarePairedDelimiter{\norm}{\lVert}{\rVert}
\DeclarePairedDelimiter{\innerprod}{\langle}{\rangle}

\DeclareMathOperator{\Vol}{Vol}
\DeclareMathOperator{\dVol}{dVol}

\DeclareMathOperator{\image}{im}

\DeclareMathOperator{\trop}{Trop}
\DeclareMathOperator{\fix}{fix}
\DeclareMathOperator{\Aut}{Aut}
\DeclareMathOperator{\PGL}{PGL}
\DeclareMathOperator{\PU}{PU}
\DeclareMathOperator{\Unitary}{U}
\DeclareMathOperator{\Isom}{Isom}

\DeclareMathOperator{\Hess}{Hess}
\DeclareMathOperator{\Log}{Log}

\newcommand{\dd}{\mathop{}\!\mathrm{d}}

\newcommand{\ZZ}{\mathbb{Z}}
\newcommand{\RR}{\mathbb{R}}
\newcommand{\CC}{\mathbb{C}}
\newcommand{\PP}{\mathbb{P}}
\renewcommand{\SS}{\mathbb{S}}

\usepackage{setspace}
\title{\centering Balanced Metrics Know About SYZ}

\author[*]{Per Berglund,}
\author[\dagger]{Tristan H\"ubsch,}
\author[\ddagger,\flat]{Vishnu Jejjala,}
\author[\sharp]{Viktor Mirjani\'c,}
\author[\sharp]{and Challenger Mishra}

\affiliation[*]{Department of Physics and Astronomy, University of New Hampshire, Durham, NH 03824, USA}
\affiliation[\dagger]{Department of Physics and Astronomy, Howard University, Washington, DC 20059, USA}
\affiliation[\ddagger]{Mandelstam Institute for Theoretical Physics, School of Physics, and NITheCS,
University of the Witwatersrand, Johannesburg, WITS 2050, South Africa}
\affiliation[\flat]{The NSF Institute for Artificial Intelligence and Fundamental Interactions (IAIFI) and Department of Physics, Northeastern University, Boston, MA 02115, USA}
\affiliation[\sharp]{Department of Computer Science and Technology, University of Cambridge, Cambridge CB3 0FD, UK}

\emailAdd{per.berglund@unh.edu}
\emailAdd{thubsch@howard.edu}
\emailAdd{v.jejjala@wits.ac.za}
\emailAdd{vvm22@cam.ac.uk}
\emailAdd{cm2099@cam.ac.uk}
\date{\today}

\abstract{
Numerical Ricci-flat metrics on Calabi--Yau manifolds are becoming increasingly accurate. However, they often lack the interpretability required to extract theoretical insights.
In this paper, we introduce a novel variant of Donaldson's algorithm  based on the Moore--Penrose pseudo-inverse that operates on the global sections of the ambient space rather than the manifold itself.
This approach allows us to use the canonical monomial basis to compute interpretable balanced metrics even at large degrees $k$.
Applying our ambient algorithm to multiple families, including the Dwork family and complete intersection Calabi--Yau manifolds, we discover that the metric parameters obey novel power laws near the Large Complex Structure Limit (LCSL).
We connect these to the Gromov--Hausdorff metric collapse predicted by the SYZ conjecture.
}

\begin{document}

\setcounter{tocdepth}{2}
\maketitle

\allowdisplaybreaks

\newpage
\input{1_Intro_revised}
\input{2_Donaldson}
\input{3_Computing_Donaldson}
\input{4_Balanced_Dwork}
\input{5_Discussion}

\appendix
\input{Appendix}

\clearpage
\bibliographystyle{JHEP}
\bibliography{refs}

\end{document}

%% file: 1_Intro_revised.tex
\section{Introduction and Summary}

While a Calabi--Yau space carries a unique Ricci-flat K\"ahler metric in each K\"ahler class for every choice of complex structure,
the existence theorem~\cite{Yau:1977ms,Yau:1978cfy} is non-constructive.
Except for the torus and the singular Kummer surface~\cite{Kachru:2020tat}, we do not know the Ricci-flat metric of a compact Calabi--Yau geometry in closed form.

A good enough numerical approximation to the Ricci-flat metric is important both physically and geometrically.
The Ricci-flat metric on the Calabi--Yau manifolds supplies solutions to string theory.
Compactifying the heterotic string on Calabi--Yau threefolds expresses the low-energy four-dimensional $\mathcal{N}=1$ theory in terms of algebraic and differential geometry~\cite{Candelas:1985en} and provides a route to the Standard Model.
While this observation established a key link between string theory as the leading candidate for quantum gravity and phenomenology in the real world, navigating the vast number of permitted geometries that comprise the string theory landscape remains an outstanding obstacle.\footnote{$473$ million four-dimensional reflexive polytopes~\cite{kreuzer_complete_2000} have up to $10^{979}$ distinct triangulations according to Ref.~\cite{MacFadden:2025ssx}, and string vacua  are widely estimated at $10^{1,500}$~\cite{Lerche:1986cx}, $10^{1,000}$~\cite{Douglas:2003um} to the oft-cited ``$10^{500}$''~\cite{Bousso:2000xa, denef_distributions_2004, Blumenhagen:2004xx, Gmeiner:2005vz}, and the more recent $10^{723}$ of Ref.~\cite{Constantin:2018xkj}.}
Exploring this landscape was historically considered intractable because physical predictions, such as normalized wavefunctions, Yukawa couplings, and higher derivative corrections, depend on knowing the manifold's Ricci-flat metric~\cite{berglund_precision_2025,Constantin:2024yxh,Fraser-Taliente:2024etl}.
In mirror symmetry, the degeneration of Ricci-flat metrics also reveals structures that are invisible at the level of topology alone.
Obtaining the metric is only part of the problem, however.
One would additionally prefer a representation whose parameters can be related to symmetries, moduli, limiting geometries, and ultimately facilitate a principled analytic description of the geometry.
The review~\cite{Berglund:2026hdh} explains the utility of the Ricci-flat metric in greater detail.

Numerical methods for Calabi--Yau metrics have advanced substantially, from Donaldson's balanced metrics and algebraic optimization to neural network approximations of K\"ahler potentials and tensor fields~\cite{Headrick:2005ch,donaldson_scalar_2001,donaldson_numerical_2005,douglas_numerical_2008,headrick_energy_2010,ashmore_machine_2020,Anderson:2020hux,douglas_numerical_2021,jejjala_neural_2022,larfors_learning_2021,berglund_machine_2023,rahman_globalcy_2026,hirst_ainstein_2025}.
These approaches increasingly make precision calculations feasible, and recent work has begun to extract compact symbolic descriptions from numerical or machine learned metrics \cite{mirjanic_symbolic_2025,lee_approximate_2025,constantin_calabi-yau_2026}.
Nevertheless, accuracy and interpretability remain in tension.
A flexible numerical ansatz can approximate the Ricci-flat metric extremely well while concealing how the answer depends on the defining equations, the symmetry group, or the position in moduli space.

Donaldson's construction offers a particularly useful setting in which to address the tension.
At level $k$, a balanced metric is encoded by a finite Hermitian matrix acting on the space of sections $H^0(X,\mathcal{O}_X(k))$, and the associated algebraic metrics converge to the desired metric as $k$ increases~\cite{donaldson_scalar_2001,donaldson_numerical_2005,sano_numerical_2006,douglas_numerical_2008}.
This convergence is notoriously slow.
For a projective variety $X\hookrightarrow\PP^n$, however, the most natural labels are the degree $k$ ambient monomials in $H^0(\PP^n,\mathcal{O}_{\PP^n}(k))$.
Once $k$ reaches the degree of the defining ideal, these monomials become linearly dependent on $X$, and the standard algorithm requires a choice of quotient basis.
That choice is mathematically harmless for the metric but obscures permutation symmetries, makes individual matrix entries basis dependent, and removes the canonical monomial labels that would otherwise make the result interpretable.
Several optimization and machine learning approaches retain the full ambient algebraic ansatz, but they replace Donaldson's fixed point iteration by a separate loss function~\cite{headrick_energy_2010,ashmore_machine_2020,gerdes_cyjax_2022,lee_approximate_2025,constantin_calabi-yau_2026}.
This leaves open the problem of preserving the canonical ambient basis within Donaldson's construction itself.

We resolve this problem by lifting Donaldson's iteration from the section space of $X$ to the full section space of its ambient variety.
If $R$ is the restriction matrix from ambient sections to sections on $X$, the Moore--Penrose pseudo-inverse selects the canonical minimal norm lift $H^\PP=(R^+)^\dagger H^X R^+$.
We then formulate an ambient $T$-map that uses the complete monomial basis together with a pseudo-inverse in place of an ordinary inverse.
We prove that this ambient map is conjugate to the usual Donaldson map whenever the restriction of sections is surjective.
It therefore computes exactly the same balanced metric on $X$, but packages the answer in a basis independent matrix whose entries retain their ambient monomial meaning.
The construction is summarized in~\cref{fig:intro}.

\begin{figure}[ht]
    \centering
    \includegraphics{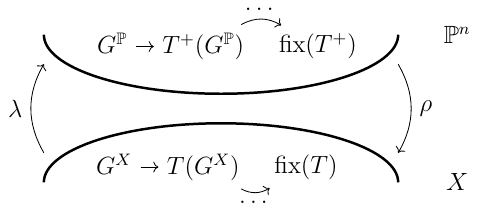}
    \caption{The usual Donaldson iteration acts on a chosen basis of $H^0(X,\mathcal{O}_X(k))$ (bottom), whereas the ambient iteration acts on the canonical monomial basis of $H^0(\PP^n,\mathcal{O}_{\PP^n}(k))$ (top).
    The Moore--Penrose pseudo-inverse removes the redundant directions generated by the defining ideal while preserving the balanced metric on $X$.}
    \label{fig:intro}
\end{figure}

The ambient formulation is especially effective near a large complex structure limit (LCSL).
The Strominger--Yau--Zaslow (SYZ) picture relates mirror symmetry to a special Lagrangian torus fibration whose fibers collapse, after an appropriate normalization, to a real affine base carrying a Monge--Amp\`ere metric~\cite{strominger_mirror_1996,gross_wilson_3tori_1997,gross_large_2000,kontsevich_homological_2001,li_stromingeryauzaslow_2022,li_geodesic_2019,li_metric_2024,li_intermediate_2025}.
The limiting metric is described locally by convex potentials related by the Legendre transform, but it is difficult to recover this structure directly from a numerical Ricci-flat metric.
The canonical ambient basis supplies precisely the missing observable: it lets us follow the coefficient of every monomial as the complex structure degenerates and as the section degree increases.

Our main results are as follows.

\begin{itemize}
\item We construct the ambient Donaldson map and prove that it is conjugate to the standard map, that its fixed point gives the canonical basis independent lift of the balanced metric, and that it extends the measure theoretic $T$-map to the degenerate measure supported on $X$.
For this lift, the essential geometric hypothesis is the surjectivity of the restriction map on sections, and we prove the required vanishing for regular complete intersections that are Calabi--Yau or almost ample in products of projective spaces.

\item We prove that the ambient iteration is equivariant under the unitary holomorphic isometries of the Ricci-flat metric and use this equivariance to reduce the numerical problem to a small set of symmetry orbits of matrix entries.
Our implementation in \texttt{cymyc}~\cite{berglund_cymyc_2024} reproduces the established $\sigma$-measure scaling on the Fermat quintic and agrees to numerical precision with an analytic solution of the Dwork zerofold.

\item For the Dwork family, we find that the normalized ambient matrix entries obey sharply defined power laws as $|\psi|\to\infty$.
The enhanced symmetry of the limiting normal crossings variety separates monomials into allowed terms, which survive in the tropical description of the base, and forbidden terms, which encode directions that disappear in the degeneration.
The finite level decay laws are numerical conjectures, but their large $k$ limit defines a convex exponent function $F$ whose Legendre transform gives a tropical potential $K_{\mathrm{trop}}$.
This construction yields the expected scaling of the base and fiber diameters and supplies a direct numerical realization of the Legendre duality underlying the SYZ picture.

\item On the Dwork torus, the exponent function is quadratic, and we obtain an explicit conjecture for the decay exponent of every allowed monomial at finite $k$.
On the Dwork K3 surface, the data are consistent with the logarithmic corrections produced by the singular affine locus, leading to a closed form first approximation for the dual Monge--Amp\`ere potential that matches the computed exponents across the simplex.
Calculations for the Cefal\'u quartics and a bicubic complete intersection provide evidence that the power law structure is not an artifact of the Dwork hypersurface presentation.

\end{itemize}

These results give two complementary interpretations of the same balanced matrix.
At fixed $k$ and finite $\psi$, its coefficients determine a \emph{continuous} approximation to the Ricci-flat K\"ahler potential via Donaldson's ansatz $K$.
Across $k$ and toward the large complex structure limit, the decay exponents of those coefficients assemble into a \emph{discrete} approximation to the Legendre dual potential $K^\vee$ on the mirror side.
In this precise sense, the balanced metrics already contain information about the SYZ collapse.

The paper is organized as follows.
In \cref{sec:better_donaldson}, we construct the ambient lift, prove its equivalence to Donaldson's original iteration, and state the conditions under which it applies to complete intersections.
In \cref{sec:computing_donaldson}, we describe the symmetry reduction and numerical implementation, validate the method, and examine balanced matrices across the Dwork family.
In \cref{sec:dwork_power_laws}, we extract the large complex structure power laws, derive their tropical and Legendre interpretation, and study the torus, K3, Cefal\'u, and complete intersection examples.
The final section compares our construction with related numerical approaches and discusses its implications and limitations.

Our code is available on GitHub at \url{https://github.com/mirjanic/cymyc} (on \verb|donaldson| branch), and the numerical data are available through Zenodo \cite{berglund_balanced_2026}.

%% file: 2_Donaldson.tex
\section{Ambient Donaldson's Algorithm}
\label{sec:better_donaldson}

In this section we use the Moore--Penrose pseudo-inverse to lift Donaldson's algorithm into a version that operates in ambient space. We show that it retains all important properties of the original, while offering better interpretability when applied to algebraic varieties.

\subsection{A Variety and Its Ambient Space}

Let $\PP^n$ be a complex projective space. To better understand the behaviour and issues with the original Donaldson's algorithm, we must distinguish between objects living on the (algebraic) variety/subspace vs on its ambient space. Thus, we mark the former with superscript ${}^X$, and the latter with ${}^\PP$.

We denote the space of global sections of degree $k$ on $\PP^n$ by $V^\PP := H^0(\PP^n, \mathcal{O}(k))$. The canonical basis of $V^\PP$ are all homogeneous monomials of degree $k$, which we denote using multi-index notation as $s^\PP := \braces*{z^\alpha}_{\abs{\alpha}=k}$. That is, the multi-indices $\alpha=\parens*{\alpha_0,\dots\alpha_n}$ satisfy $\sum \alpha_i=k$. The number of these monomials is $N^\PP := \dim V^\PP = \binom{n+k}{n}$. 

Let $X \hookrightarrow \PP^n$ be a Calabi--Yau hypersurface of dimension $n-1$ with a defining polynomial $Q$ of degree $n+1$, for $X:=\{Q{=}0\}\subset\PP$ to be Calabi--Yau. Let $V^X := H^0(X, \mathcal{O}(k))$ denote the space of global sections on $X$, with a basis $s^X$ of dimension $N^X$. For $k\le n$ we have $V^\PP\cong V^X$, but for $k>n$ the defining polynomial $Q$ introduces linear dependencies that reduce the dimensionality of $V$.

This owes to the fact that the sections $s^X$ are really elements of the quotient ring $\mathbb{C}[z]/\mathcal{I}$, where $z\in H^0(\PP,\mathcal{O}(1))$ are the homogeneous variables on the ambient space $\PP^n$, and $\mathcal{I}$ is the ideal generated by the defining polynomial $Q\colon X \hookrightarrow \PP^n$
: $s^X \cong [s^\PP \pmod Q]$.
Clearly, $\mathcal{I}$ starts imposing relations on $\mathbb{C}[z]$ only 
for $k\geqslant\deg[Q]=n{+}1$.

The above generalises to different ambient spaces, such as complete intersection Calabi--Yau spaces (CICYs), where the ideal $\mathcal{I}$ is generated by a system of equations $f\colon X \hookrightarrow \PP^{n_1}{\times}{\cdots}{\times}\PP^{n_r}$. This gets even more varied on weighted projective spaces, and on a general toric variety the $z$'s are the Cox variables.

\subsection{Restriction Map}

The relationship between $V^\PP$ and $V^X$ is controlled by the following short exact sequence of sheaves (see also \cite{douglas_numerical_2008})
\begin{equation}\label{eq:sheaf_ses}
    0\to \mathcal{O}_{\PP^n}(k{-}n{-}1) \xrightarrow{\cdot Q} \mathcal{O}_{\PP^n}(k) \to \mathcal{O}_X(k) \to 0,
\end{equation}
which essentially encodes the fact that multiplying any section by $Q$ and restricting to Calabi--Yau where $Q=0$ will annihilate the original. Taking the cohomology we obtain the short exact sequence
%
\begin{equation}
    \begin{aligned}
    0\to & H^0(\PP^n, \mathcal{O}(k{-}n{-}1)) \xrightarrow{\cdot Q} H^0(\PP^n, \mathcal{O}(k)) \xrightarrow{\rho}
    \tikz[remember picture,baseline=(A.base),inner sep=0pt,outer sep=0pt]{\node (A) {$H^0(X,\mathcal O(k))$};}\\
    &
    \tikz[remember picture,baseline=(B.base),inner sep=0pt,outer sep=0pt]{\node (B) {$H^1(\PP^n,\mathcal O(k{-}n{-}1))$};}=0
    \end{aligned}
    \begin{tikzpicture}[remember picture,overlay]
        \def\offset{1.5mm}
        \def\looseness{2.2}
        \coordinate (M) at ($(A.east)!0.5!(B.east)$);
        \coordinate (A1) at ($(A.east)+(1mm,0)$);
        \coordinate (A2) at ($(A1)+(\offset,0)$);
        \coordinate (B1) at ($(B.west)+(-0.5mm,0)$);
        \coordinate (B2) at ($(B1)+(-\offset,0)$);
        \draw[semithick,-{Computer Modern Rightarrow}]
        (A1) -- (A2)
        to[out=0,in=0,looseness=\looseness] (A2 |- M)
        to[out=180,in=0] (B2 |- M)
        to[out=180,in=180,looseness=\looseness] (B2)
        -- (B1);
    \end{tikzpicture}
\end{equation}
where $\rho: V^\PP \to V^X$ is a linear restriction map that satisfies
\begin{equation}
  \ker\rho
  = Q \cdot H^0(\PP^n, \mathcal{O}(k{-}n{-}1))
  \cong H^0(\PP^n, \mathcal{O}(k{-}n{-}1)).
\end{equation}
Since these are vector spaces, the sequence splits, allowing us to introduce a lift map $\lambda$ and write
\begin{equation}
    0 \to \ker\rho \to V^\PP \overset{\rho}{\underset{\lambda}{\rightleftarrows}} V^X \to 0.
\end{equation}
There could be infinitely many $\lambda$s such that $\rho\circ\lambda=\text{id}$. Nevertheless, we will show in \cref{sec:lift_t_map} that the unique adjoint of $\rho$ is the correct choice of $\lambda$ that one should take.\footnote{It is worth noting that summing up these vector spaces over $k$ yields graded rings. For $V^\PP$ this is the polynomial ring $\CC[z]$, while for $V^X$ it is the quotient ring $\CC[z]/\mathcal{I}$. While $\rho$ respects this structure and is hence a ring homomorphism, our lift $\lambda$ is not. Exploring this further is beyond the scope of our paper as it has no immediate bearing on Donaldson's algorithm.}

Since we have a canonical basis $s^\PP$ on $V^\PP$, and some chosen basis $s^X$ on $V^X$, we can introduce the \emph{restriction matrix} $R$ of size $N^\PP\times N^X$ as the matrix representation of $\rho$ acting on the bases, satisfying
\begin{equation}
\label{eq:restriction}
    s^\PP|_X :=\rho(s^\PP)= R s^X.
\end{equation}
Because $\rho$ is surjective, $R$ has full column rank $N^X$. In fact, we have $V^X\cong V^\PP / \ker\rho$, so we conclude 
\begin{equation}
    N^X = \dim V^\PP - \dim \ker\rho = \binom{n + k}{n} - \binom{k-1}{n},
\end{equation}
where $\binom{a}{b}=0$ if $a<b$.

\subsection{Donaldson's Algorithm}

Let $H>0$ be a positive-definite matrix. It defines a K\"ahler potential on $\PP^n$,
\begin{equation}\label{eq:algebraic_ansatz}
    K = \frac{1}{k\pi}\log \big(
    \parens{s^\PP}^\dag H s^\PP \big),
\end{equation}
an ansatz called the `algebraic potential.' For example, some choice of $H$ will reproduce the Fubini--Study potential. Furthermore, this potential trivially pulls back to any variety $X\hookrightarrow\PP^n$ such as the Calabi--Yau manifold. For these pullbacks it is sufficient that $H>0$ on the variety $X$, and in our constructions below $H$ will only be positive semidefinite in the ambient space.

Donaldson's insight is that we can try to find an optimal $H$ that approximates any measure on $\PP^n$, not necessarily the uniform Fubini--Study measure. Formally, let $\nu$ be a positive Radon measure on $\PP^n$. \cite{donaldson_scalar_2001,donaldson_numerical_2005} defines the $T$-map
\begin{equation}
    T_\nu\parens*{G} = \frac{N^\PP}{\Vol_{\dd\nu}} \int_{\PP^n} \frac{s^\PP \parens{s^\PP}^\dag}{\parens{s^\PP}^\dag G^{-1} s^\PP} \dd\nu,
\end{equation}
and proposes that under some non-degeneracy conditions on $\nu$, this map is a contraction. Therefore, for each degree $k$, $T_\nu$ has a fixed point to which any initial matrix will converge. 

The $T_\nu$-map is homogeneous and satisfies $T_\nu(\lambda G) = \lambda T_\nu(G)$, so the fixed point of $T_\nu$ is unique only up to scaling. We will denote this fixed point with $\fix(T_\nu)$ and assume that the scaling degree of freedom is implied. In numerical computations, we will eliminate this symmetry by normalising a particular entry in the matrix to 1. 

Setting $H=\fix(T_\nu)^{-1}$, we obtain the potential associated with this fixed point, called the \emph{$\nu$-balanced potential} of degree $k$ (or at level $k$). Furthermore, Donaldson showed that in the $k\to\infty$ limit, $\nu$-balanced metrics converge to the volume given by $\nu$.

We would like to recover the Ricci-flat metric on $X$ using the flat measure given by the holomorphic volume form $\Omega$. While this measure is valid on $X$, it is degenerate on $\PP^n$, failing the necessary preconditions and making the integral singular. The classical solution is to restrict to $X$ and iterate the map there. Thus, one has
\begin{equation}
\label{eq:donaldson}
    T(G) = \frac{N^X}{\Vol_\Omega} \int_X \frac{s^X \parens*{s^X}^\dag}{\parens*{s^X}^\dag G^{-1} s^X} \dVol_\Omega,
\end{equation}
and by setting $H^X=\fix(T)^{-1}$ one can approximate the Ricci-flat potential with
\begin{equation}
    \frac{1}{k\pi}\log \big( (s^X)^\dag H^X s^X \big).
\end{equation}
To find this fixed point it is sufficient to iterate \cref{eq:donaldson}. This clearly requires specifying the basis $s^X$. A natural choice is to use the canonical basis of all monomials of degree $k$ in $\PP^n$, as they simplify computation and respect the subvariety automorphisms. However, as we have already seen, this choice only works for $k\le n$. In turn, any choice of basis with $k\geqslant n{+}1=\deg[Q]$ will break the symmetries of the subvariety $X:=\{Q{=}0\}$. Consequently, the $H^X$ matrix will depend on the chosen basis, reducing interpretability.

In the following section we will show that this effect is \emph{reversible}, allowing us to always use the canonical basis $s^\PP$, obtaining highly symmetric and interpretable balanced metrics even when $k>n$.

\subsection{The Canonical Lift \texorpdfstring{$H^\PP$}{}}

We seek a positive-semidefinite matrix $H^\PP$ of size $N^\PP \times N^\PP$ acting on the full ambient basis $s^\PP$ that reproduces the physical metric on $X$. This requirement implies that
\begin{equation}
    \parens*{s^\PP}^\dag H^\PP s^\PP\big|_X = \parens*{s^X}^\dag H^X s^X
\end{equation}
should hold pointwise on the manifold, where $H^X$ is the balanced matrix from Donaldson's algorithm of size $N^X \times N^X$. Substituting \cref{eq:restriction}, we obtain the algebraic constraint on $H^\PP$:
\begin{equation}
\label{eq:h_full_condition}
    R^\dag H^\PP R = H^X.
\end{equation}
While the solution to \cref{eq:h_full_condition} is not unique due to $\ker \rho$ being non-trivial when $k>n$, there exists a unique \emph{best} solution, as the following proposition shows.

\begin{prop}
\label{prop:h_full}
Define $H^\PP := \parens{R^+}^\dag H^X R^+ $, where $R^+:=(R^\dag R)^{-1} R^\dag$ denotes the Moore--Penrose pseudo-inverse of $R$. Then, this $H^\PP$
\begin{itemize}
    \item is a solution to \cref{eq:h_full_condition},
    \item is, in fact, a solution with minimal Frobenius norm,
    \item is in the column space of $R$, meaning $\forall v\in \ker\rho$, $v^\dag H^\PP=H^\PP v=0$.
\end{itemize}
\end{prop}
\begin{proof}
    Since $R$ has full column rank, $R^+$ is its left inverse, meaning $R^+ R=I$, and the first point immediately follows. Second and third points follow from well-known properties of the pseudo-inverse.
\end{proof}

The last condition can be interpreted as saying that $H^\PP$ respects $\rho$ and assigns no \textquote{energy} to directions that are orthogonal to $X$. However, the most important property of $H^\PP$ is that it extends the balanced metric to $\mathbb{P}^n$ in the most natural way possible.

\begin{theorem}
\label{thm:hpp_inv}
The matrix $H^\PP$ is independent of the choice of basis for $V^X$.
\end{theorem}
\begin{proof}
Let $s^X$ be a basis for $V^X$ with associated Donaldson matrix $H^X$ and restriction matrix $R$. Consider a change of basis to $\tilde{s}^X = M s^X$, where $M \in GL(N^X, \mathbb{C})$.

Under this transformation, the restriction matrix $\tilde{R}$ transforms as:
\begin{equation}
    s^\PP|_X = R s^X = R ( M^{-1} \tilde{s}^X) =  (R M^{-1}) \tilde{s}^X.
\end{equation}
Thus, $\tilde{R} = R M^{-1}$.

Since the balanced metric is unique, it is independent of the basis. Therefore, one has $\parens{s^X}^\dag H^X s^X = \parens{\tilde{s}^X}^\dag \tilde{H}^X \tilde{s}^X$ for all $\tilde{s}^X$. This implies $H^X = M^\dag \tilde{H}^X M$, and thus
\begin{equation}
    \tilde{H}^X = (M^{-1})^\dag H^X M^{-1}. 
\end{equation}
Finally, we compute
\begin{equation}
\begin{aligned}
    \tilde{H}^\PP &= \parens{\tilde{R}^+}^\dag \tilde{H}^X \tilde{R}^+ \\ 
    &= \parens{R M^{-1}}^{+\dag} \parens*{(M^{-1})^\dag H^X M^{-1}} \parens{R M^{-1}}^+\\ 
    &=  \big(\parens{R^+}^\dag M^{\dag}\big) \parens{M^{-1}}^\dag H^X M^{-1} \parens{M R^+} \\ 
    &= \parens{R^+}^\dag H^X R^+ \\
    &= H^\PP. 
\end{aligned}
\end{equation}
Therefore, $H^\PP$ is independent of the basis as claimed.
\end{proof}

\subsection{Lifting the \texorpdfstring{$T$}{T}-map}
\label{sec:lift_t_map}

\cref{thm:hpp_inv} lets us run Donaldson's algorithm in any basis to obtain $H^X$ and then move to the basis invariant $H^\PP$. However, we can instead modify Donaldson's algorithm to operate in the ambient space directly, and thus avoid ever having to specify $s^X$. We define this projective T-map as 
\begin{equation}
\label{eq:projective_donaldson}
    T^+_X\parens*{G} := \frac{N^X}{\Vol_\Omega} \int_X \frac{s^\PP \parens*{s^\PP}^\dag}{\parens*{s^\PP}^\dag G^+ s^\PP} \dVol_\Omega.
\end{equation}
This is identical to \cref{eq:donaldson}, except one now uses ambient basis $s^\PP$ instead of $s^X$ and pseudo-inverse $G^+$ instead of $G^{-1}$. The only dependence on geometry of $X$ remains in the integral and in the normalisation constant $N^X / \Vol_\Omega$. We claim that \cref{eq:projective_donaldson} computes the same object as the original $T$-map.

Specifically, let the lift map $\lambda$ have matrix representation $R^+$ to form an adjoint pair. For $H$ this results in $\lambda(H)=\parens{R^+}^\dag H R^+$, as introduced above. Meanwhile, $G$ is the inverse of $H$ resulting in dual action $\lambda(G) = R G R^\dag$. We propose that $\lambda$ is a \emph{conjugacy} between maps $T$ and $T^+_X$.
\begin{theorem}\label{thm:t_is_tpx}
    $\lambda \circ T = T^+_X \circ \lambda$.
\end{theorem}
\begin{proof}
We directly compute 
{
\allowdisplaybreaks
\begin{equation}
\begin{aligned}
    T^+_X(\lambda(G)) &= \frac{N^X}{\Vol_\Omega} \int_X \frac{s^\PP \parens*{s^\PP}^\dag}{\parens*{s^\PP}^\dag \parens*{ R G R^\dag}^{+} s^\PP} \dVol_\Omega\\
    &= \frac{N^X}{\Vol_\Omega} \int_X \frac{R s^X \parens{R s^X}^\dag}{\parens*{R s^X}^\dag (R^\dag)^+ G^{-1} R^+ (R s^X)} \dVol_\Omega\\
    &= R \frac{N^X}{\Vol_\Omega} \int_X \frac{s^X \parens{s^X}^\dag}{\parens*{s^X}^\dag \parens{R^+R }^\dag G^{-1} \parens{R^+ R} s^X} \dVol_\Omega R^\dag\\
    &= R \frac{N^X}{\Vol_\Omega} \int_X \frac{s^X \parens{s^X}^\dag}{\parens*{s^X}^\dag G^{-1} s^X} \dVol_\Omega R^\dag\\
    &= \lambda(T(G)).
\end{aligned}
\end{equation}
}
in which $R^+ R = I$ allows the inner terms to collapse.
\end{proof}
Let $\mathcal{C}_X := \braces{G^X\ |\ G^X > 0}$ be the cone of positive-definite $N^X\times N^X$ matrices that $T$ operates on, and $\lambda(\mathcal{C}_X) := \braces*{RG^X R^\dag\ |\ G^X > 0}$ be its lift under $\lambda$. \cref{thm:t_is_tpx} allows us to easily translate properties of $T$ acting on $\mathcal{C}_X$ to that of $T^+_X$ acting on $\lambda(\mathcal{C}_X)$.
\begin{coro}
    The $T^+_X$-map is contracting and has a unique fixed point in $\lambda(\mathcal{C}_X)$. 
\end{coro}
\begin{coro}
    $\lambda\parens{\fix(T)} = \fix\parens{T^+_X}$.
\end{coro}
These claims can be strengthened even further. Define $\mathcal{C}^+_X := \braces*{G^\PP\ |\ R^\dag (G^\PP)^+ R > 0}$. One can check that if $G^X\in\mathcal{C}_X$ then $G^\PP=RG^XR^\dag\in\mathcal{C}^+_X$ and we have $\lambda(\mathcal{C}_X) \subsetneq \mathcal{C}^+_X$. Furthermore, $\mathcal{C}^+_X$ does not actually depend on the basis even though it contains $R$ in the definition.
\begin{prop}
    Let $G^\PP\in\mathcal{C}^+_X$. Then, $T^+_X(G^\PP) = R T\parens*{\parens*{R^\dag (G^\PP)^+ R}^{-1}} R^\dag$.
\end{prop}
\begin{proof}
Note that right-hand side is well defined based on the assumptions. Let $G^X := \parens*{R^\dag (G^\PP)^+ R}^{-1}$.
\begin{equation}
    \begin{aligned}
        R T(G^X) R^\dag &= \lambda(T(G^X)) \\
        &= T^+_X(\lambda(G^X)) \\
        &= \frac{N^X}{\Vol_\Omega} \int_X \frac{s^\PP \parens*{s^\PP}^\dag}{\parens*{s^\PP}^\dag \parens*{R G^X R^\dag}^+ s^\PP} \dVol_\Omega \\
        &= \frac{N^X}{\Vol_\Omega} \int_X \frac{s^\PP \parens*{s^\PP}^\dag}{\parens*{s^\PP}^\dag \parens*{(R R^+)^\dag (G^\PP)^+ R R^+} s^\PP} \dVol_\Omega \\
        &= \frac{N^X}{\Vol_\Omega} \int_X \frac{s^\PP \parens*{s^\PP}^\dag}{\parens*{s^\PP}^\dag (G^\PP)^+ s^\PP} \dVol_\Omega \\
        &= T^+_X(G^\PP).
    \end{aligned}
\end{equation}
We used $R R^+ s^\PP = s^\PP$ to simplify the denominator because $s^\PP$ is in the column space of $R$ on $X$.
\end{proof}
Therefore, in the first step, $T^+_X$ will project $G^\PP\in\mathcal{C}^+_X$ onto $\lambda(C_X)$.
\begin{coro}
    The $T^+_X$-map is contracting and has a unique fixed point in $\mathcal{C}^+_X$.
\end{coro}

Since $\lambda(\mathcal{C}_X) \subsetneq \mathcal{C}^+_X$, this is strictly more general. This does not impact convergence, but it does offer more flexibility when choosing the initialisation.

We conclude that iterating $T^+_X$ until convergence and taking the pseudo-inverse at the end to obtain $\fix(T^+_X)^+$ is another valid way to obtain the balanced metric on $X$. Thus, we can compute $H^\PP$ directly and obtain balanced metrics in a way that bypasses $s^X$. Finally, $T^+_X$-map ties back to the measure-theoretic formulation of Donaldson's algorithm. 

\begin{prop}
    For a positive-definite matrix $G$, $T^+_X$ extends $T_\nu$ up to a constant:
    \begin{equation}
        \lim_{\dd\nu\to\dVol_\Omega} T_\nu(G) = \frac{N^X}{N^\PP} T^+_X(G).
    \end{equation}
\end{prop}
\begin{proof}
    Let $G$ be a positive definite matrix. It is non-singular, so $G^+ = G^{-1}$. Furthermore,  $(s^\PP)^\dag G^{-1} s^\PP$ is strictly positive everywhere on $\PP^n$. Therefore, the integrand is well behaved and continuous.
    
    By weak-* convergence $\dd\nu \rightharpoonup^* \dVol_\Omega$, we can rewrite the integration domain to $X$ since $\dVol_\Omega$ localises to $X$, and the result immediately follows.
\end{proof}

Similar convergence statements can also be shown for the fixed points. In total, we can think of our $T^+_X$-map as a way of extending original $T_\nu$ to degenerate measures.

\subsection{Generalisations and Limitations}

So far, we provided proof for hypersurfaces in $\PP^n$. It is natural to ask whether these results hold more generally. Main condition for correctness is the surjectivity of the restriction map $\rho$. This was true for $\PP^n$ because higher cohomology $H^1\big(\PP^n, \mathcal{O}(k-n-1)\big)$ trivially vanished, but now we move to more general constructions.

Consider a regular complete intersection manifold $X \hookrightarrow A = \PP^{n_1}\times \dots \times\PP^{n_r}$ with $m$ defining equations and a configuration matrix
\begin{equation}
    \bracks*{ \begin{array}{c|ccc}
        \PP^{n_1} & d_{1,1} & \cdots & d_{m,1} \\
        \vdots & \vdots & \ddots & \vdots \\
        \PP^{n_r} & d_{1,r} & \cdots & d_{m,r} \\
    \end{array} }
\end{equation}
We assume the polynomials are independent and intersect transversely, and the manifold has codimension $m$. Smoothness is not strictly necessary and depends on choice of complex structure moduli, which are related polynomial coefficients. Furthermore, we assume that all degrees are non-negative, so $d_{i,j}\ge0$. Total degrees per ambient space $\PP^{n_i}$ are $D_i = \sum_j d_{j,i}$. Following \cite{hubsch_calabi-yau_2024}, if $D_i = n_i+1$ for all $i$ then $X$ is Calabi--Yau, and if $D_i \le n_i+1$ then it is almost ample. Otherwise, we say $X$ is of general type. We note that a space not being Calabi--Yau is not an issue by itself since Donaldson's algorithm can be applied to any measure.

Let $\vec k=(k_1,\dots,k_r)$ be the degrees of sections in the algorithm with respect to each projective space in $A$. Clearly, the algorithm requires at least $k_i\ge0$. If $k_i=0$ for some $i$ then the resulting metric would be degenerate, as it would not receive contributions from $\PP^{n_i}$. Therefore, we assume $k_i\ge1$ for all $i$.

The short exact sequence from \cref{eq:sheaf_ses} is now replaced by the Koszul complex
\begin{equation}
    0 \to \parens*{\bigwedge^m\mathcal{E}}(\vec k) \to \parens*{\bigwedge^{m-1}\mathcal{E}}(\vec k) \to \dots \to  \mathcal{E}(\vec k) \to \mathcal{O}_A(\vec k) \to \mathcal{O}_X(\vec k) \to 0,
\end{equation}
where $\mathcal{E}=\bigoplus^m_{i=1}\mathcal{O}_A(-d_{i,1},-d_{i,2},\dots)=\bigoplus^m_{i=1}\mathcal{O}_A(-\vec d_i)$. This can be split into
\begin{align}
    0 \to \parens*{\bigwedge^m\mathcal{E}}(\vec k) \to \dots \to  \mathcal{E}(\vec k) &\to \mathcal{I}_X(\vec k) \to 0 \label{eq:sheaf_koszul_a} \\
    0 &\to \mathcal{I}_X(\vec k) \to \mathcal{O}_A(\vec k) \to \mathcal{O}_X(\vec k) \to 0, \label{eq:sheaf_koszul_b}
\end{align}
where $\mathcal{I}_X$ is the ideal sheaf of $X$.

\begin{theorem}
    If $X$ is regular CICY or almost ample manifold, $k_i\ge1$ for all $i$ is sufficient for cohomological obstructions to vanish.
\end{theorem}
\begin{proof}
 We wish to prove that $H^1(A,\mathcal{I}_X(k))$ vanishes under above conditions.
    
    The general term in \cref{eq:sheaf_koszul_a} is obtained by taking symmetric sums of degrees $\vec d_i$
    \begin{equation}
        \parens*{\bigwedge^p\mathcal{E}}(\vec k)=\bigoplus_{x}\mathcal{O}_A(\vec k - \vec d_{x_1} - \dots - \vec d_{x_p}).
    \end{equation}
    These terms are always bounded from below with $\vec k-\vec D$ (with comparisons performed elementwise). Since $\vec D\le \vec n+1$ and $\vec k\ge1$, they are in fact always at least $-\vec n$. Using K\"unneth formula \cite[Proposition 9.2.4]{kempf_general_1993} on each $\mathcal{O}_A$ inside $\bigoplus_x$ we obtain 
    \begin{equation}
        H^q\parens*{A,\parens*{\bigwedge^p\mathcal{E}}(\vec k)}
        =\bigoplus_{x}\parens*{\bigoplus_{\sum q_i=q}\bigotimes_{i} H^{q_i}\parens*{\PP^{n_i}, \mathcal{O}(\underbrace{k_i - d_{x_1,i} - \dots - d_{x_p,i}}_{t})}}.
    \end{equation}
    By explicitly computing the cohomology \cite[III, Theorem 5.1]{hartshorne_algebraic_1977}, the intermediate cohomologies $H^{q_i}(\PP^{n_i},\mathcal{O}(t))$ always vanish, and top cohomology $H^{n_i}$ vanishes if $t\ge -n_i$. This is always satisfied in our case since we showed $k_i-\sum d_{x_j,i}\ge -n_i$. Therefore, every non-zero cohomology will vanish.

    Now, for $q>0$, at least one $q_i$ must be strictly positive. Therefore, there is at least one cohomology $H^{q_i}$ inside $\bigotimes$ that vanishes by above, and it annihilates the entire tensor product. From this, we conclude
    \begin{equation}
        H^q\parens*{A,\parens*{\bigwedge^p\mathcal{E}}(\vec k) } = 0 \text{ for all $q>0$ and all $p$}.
    \end{equation}
    We will conclude by breaking \cref{eq:sheaf_koszul_a} into short exact sequences and chasing $H^1$ up the chain to show that it vanishes. Let $Z_p=\image\parens*{\parens{\bigwedge^p\mathcal{E}}(\vec k)\to\parens{\bigwedge^{p-1}\mathcal{E}}(\vec k)}$, along with $Z_1=\mathcal{I}_X(\vec k)$ and $Z_m=\parens{\bigwedge^m\mathcal{E}}(\vec k)$. This yields sequences
    \begin{equation}
        0 \to Z_p \to \parens*{\bigwedge^{p-1}\mathcal{E}}(\vec k) \to Z_{p-1} \to 0.
    \end{equation}
    Taking the cohomology we obtain
\begin{equation}
    \begin{aligned}
    \dots \to H^q\bigg(A,\bigg(\bigwedge^{p-1}\mathcal{E}\bigg)(\vec k)\bigg) \to 
     &
    \tikz[remember picture,baseline=(A.base),inner sep=0pt,outer sep=0pt]{\node (A) {$H^q(A,Z_{p-1})$};}\\[-4mm]
    &
    \tikz[remember picture,baseline=(B.base),inner sep=0pt,outer sep=0pt]{\node (B) {$H^{q+1}(A,Z_p)$};} \to 
      H^{q+1}\bigg(A,\bigg(\bigwedge^{p-1}\mathcal{E}\bigg)(\vec k)\bigg)\to \dots
    \end{aligned}
    \begin{tikzpicture}[remember picture,overlay]
        \def\offset{1.5mm}
        \def\looseness{2.2}
        \coordinate (M) at ($(A.east)!0.5!(B.east)$);
        \coordinate (A1) at ($(A.east)+(1mm,0)$);
        \coordinate (A2) at ($(A1)+(\offset,0)$);
        \coordinate (B1) at ($(B.west)+(-0.5mm,0)$);
        \coordinate (B2) at ($(B1)+(-\offset,0)$);
        \draw[semithick,-{Computer Modern Rightarrow}]
        (A1) -- (A2)
        to[out=0,in=0,looseness=\looseness] (A2 |- M)
        to[out=180,in=0] (B2 |- M)
        to[out=180,in=180,looseness=\looseness] (B2)
        -- (B1);
    \end{tikzpicture}
\end{equation}
As we previously established, the leftmost and rightmost cohomologies vanish. Thus,
\begin{equation}
    H^q(A,Z_{p-1}) \cong H^{q+1}(A,Z_p),
\end{equation}
and
\begin{equation}
    H^1(A,\mathcal{I}_X(\vec k))\cong H^2(A,Z_2)\cong\dots\cong H^m\parens*{A,\parens*{\bigwedge^m\mathcal{E}}(\vec k)}=0.
\end{equation}
Finally, taking the cohomology of \cref{eq:sheaf_koszul_b} we obtain
\begin{equation}
    0 \to H^0(A,\mathcal{I}_X(\vec k)) \to H^0(A,\mathcal{O}(\vec k)) \to H^0(X,\mathcal{O}(\vec k)) \to H^1(A,\mathcal{I}_X(\vec k)) = 0.
\end{equation}
Therefore, the map $H^0(A,\mathcal{O}(\vec k)) \to H^0(X,\mathcal{O}(\vec k))$ is surjective.
\end{proof}
For varieties of general type, Serre's vanishing theorem \cite[III, Theorem 5.2]{hartshorne_algebraic_1977} implies that the higher cohomologies will always trivialise eventually, that is for $\vec k\gg1$. The same argument applies to weighted projective spaces since they, too, are projective schemes. However, unlike CICY and almost ample cases, there is no nice bound on $\vec k$ for when this happens. That being said, this is perhaps not necessary since convergence happens when $\vec k\to\infty$ anyway.

A much more serious issue is that there are general toric varieties that are \emph{not} projective. For them, $H^1$ does not have to vanish even at arbitrarily large level $k$, creating an unsalvageable obstruction. However, this is more broadly a problem for Donaldson's algorithm itself. On such manifolds the Tian--Yau--Zelditch--Lu expansion \cite{lu_lower_2000,zelditch_szego_1998} fails, so Donaldson's iteration scheme has no convergence guarantees at all.

Our lifting cannot improve upon this, so the point at which our work fails is the point where the original convergence guarantees itself fail. In any case, we proceed to restrict ourselves to CICYs in products of ordinary (non-weighted) projective spaces.

Finally, note that using the canonical ambient basis $s^\PP$ (instead of the hypersurface-specific basis $s^X$) clearly extends the above calculations over entire deformation families of defining sections, $Q=Q(\psi)$. This affords tracking $\psi$-dependence of the balanced metric more directly. We will explore this complex structure dependence in the following section.

%% file: 3_Computing_Donaldson.tex
\section{Computing Ambient Balanced Metrics}\label{sec:computing_donaldson}

In this section, we compute balanced metrics in ambient space on the Dwork family. We show that the numerical properties of our approach perfectly match existing known behaviour, suggesting that we are indeed learning the same balanced metrics. We also explore in detail the main distinguishing feature of our algorithm, namely the presence of all homogeneous monomials. We show that this simplifies reasoning about symmetries and allows for making quantitative statements about balanced metrics.

\subsection{Symmetry Handling}

We implement the $T^+_X$ iteration in \verb|cymyc| by \cite{berglund_cymyc_2024}. For the convergence to be well behaved, $G^\PP$ needs to be initialised to the row space of $R$, and the easiest way is to set $G^\PP = I$. Instead, we choose $G^\PP$ so that $H^\PP_{\alpha\bar\alpha} = \frac{1}{\alpha_0!\,\alpha_1!\,\cdots\,\alpha_n!}$ where $\alpha$ are multinomial coefficients, because this initialises the algorithm to the Fubini--Study metric:
\begin{equation}
\begin{aligned}
K_\mathrm{init} &= \frac{1}{k\pi}\log(s^\PP)^\dag H^\PP s^\PP = \frac{1}{k\pi}\log\parens*{\frac{1}{k!}\parens*{\sum \abs{z_i}^2}^k} = \frac{1}{\pi}\log\sum\abs{z_i}^2 + c\\ 
 &= K_\mathrm{FS} + c.
\end{aligned}
\end{equation}
This choice is motivated by evidence that on some toy manifolds, like the Fermat family, the pulled back Fubini--Study metric is exceptionally close to being Ricci-flat. Furthermore, this initialisation is positive-definite everywhere so it satisfies $G^\PP\in\mathcal{C}^+_X$ as required.

The primary obstruction to scaling Donaldson-based approaches is that the number of sections $N_\PP$ grows polynomially in $k$, but we can mitigate this. We operate on manifolds with large isometry group of the Ricci-flat metric, $\Isom(X, g_\mathrm{flat})$. This is relevant to us because many of these symmetries are shared by balanced metrics for every $k$. 

We will prove this for hypersurfaces in $\PP^n$ for which all symmetries are unitary. Let $\Aut(G^\PP) = \braces*{A\in\PGL(N^\PP,\CC)\colon AG^\PP A^\dag = G^\PP}$, and let $\Isom_\text{hol}$ be the holomorphic subgroup of $\Isom$.
\begin{prop}\label{prop:T+X-equivariant}
    $T^+_X$ is equivariant under $\Isom_\mathrm{hol}(X_\psi, g_\mathrm{CY}) \cap \PU(n)$.
\end{prop}
\begin{proof}
Let $f\in\Isom_\mathrm{hol}(X_\psi, g_\mathrm{CY}) \cap \PU(n)$ be an isometry of the Ricci-flat metric. It induces an action on monomials $F\in\Unitary(N^\PP)$ via the Veronese embedding, satisfying $F s^\PP(z) = s^\PP(f(z))$. We will show $F T^+_X(G) F^\dag = T^+_X(F G F^\dag)$.
\begin{small}
\begin{align}
        F T^+_X(G) F^\dag &= \frac{N^X}{\Vol_\Omega} \int_X \frac{F s^\PP(z) \left(F s^\PP(z)\right)^\dag}{s^\PP(z)^\dag G^+ s^\PP(z)} \, \mathrm{dVol}_\Omega(z) 
        && \text{(expand $T^+_X$)}\nonumber\\ 
        &= \frac{N^X}{\Vol_\Omega} \int_X \frac{s^\PP(f(z)) s^\PP(f(z))^\dag}{\left(F^{-1} s^\PP(f(z))\right)^\dag G^+ \left(F^{-1} s^\PP(f(z))\right)} \, \mathrm{dVol}_\Omega(z) 
        && \text{(by $F\circ s^\PP = s^\PP \circ f$)}\nonumber\\
        &= \frac{N^X}{\Vol_\Omega} \int_X \frac{s^\PP(f(z)) s^\PP(f(z))^\dag}{s^\PP(f(z))^\dag (F G F^\dag)^+ s^\PP(f(z))} \, \mathrm{dVol}_\Omega(z) 
        && \text{(group $(FGF^\dag)^+$)}\nonumber\\
        &= \frac{N^X}{\Vol_\Omega} \int_X \frac{s^\PP(u) s^\PP(u)^\dag}{s^\PP(u)^\dag (F G F^\dag)^+ s^\PP(u)} \, (f^{-1})^*\mathrm{dVol}_\Omega(u) 
        && \text{(set $u=f(z)$)}\nonumber\\
        &= \frac{N^X}{\Vol_\Omega} \int_X \frac{s^\PP(u) s^\PP(u)^\dag}{s^\PP(u)^\dag (F G F^\dag)^+ s^\PP(u)} \, \mathrm{dVol}_\Omega(u) && \text{(pullback)}\nonumber\\
        &= T^+_X(F G F^\dag).
\end{align}
\end{small}%
In the above, $(F^{-1})^\dag G^+ F^{-1} = (F G F^\dag)^+$ holds because $F$ is unitary. Furthermore, $(f^{-1})^*\mathrm{dVol}_\Omega=\mathrm{dVol}_\Omega$ holds because $f$ is an isometry of $g_\mathrm{CY}$ and thus preserves the volume.
\end{proof}
\begin{prop} The fixed point $G^\PP=\fix(T^+_X)$ satisfies 
    $\Aut(G^\PP) \ge \Isom_\mathrm{hol}(X_\psi, g_\mathrm{CY}) \cap \PU(n)$.
    \begin{proof}
        By definition, $G^\PP$ is the unique fixed point up to scaling that satisfies $G^\PP=T^+_X(G^\PP)$. Applying the symmetry $F$ we obtain $FG^\PP F^\dag = F T^+_X(G^\PP) F^\dag$. By \cref{prop:T+X-equivariant}, this also equals $T^+_X(F G^\PP F^\dag)$. Therefore, $FG^\PP F^\dag$ is also the fixed point, so we conclude $G^\PP = FG^\PP F^\dag$.
    \end{proof}
\end{prop}
\begin{coro}
    $\Isom_\mathrm{hol}(X_\psi, g_\mathrm{balanced}) \ge \Isom_\mathrm{hol}(X_\psi, g_\mathrm{CY})  \cap \PU(n)$.
\end{coro}
Since $g_\text{balanced}$ approximates $g_\text{CY}$, we have $\Isom_\mathrm{hol}(X_\psi, g_\mathrm{balanced}) \le \Isom_\mathrm{hol}(X_\psi, g_\mathrm{CY})$ for sufficiently large $k$. On the other hand, we will only work with manifolds where all symmetries are unitary, so we can conclude that $\Isom_\mathrm{hol}(X_\psi, g_\mathrm{balanced}) \ge \Isom_\mathrm{hol}(X_\psi, g_\mathrm{CY})$ for all $k$. 

Therefore, balanced and Ricci-flat metrics will have the same isometries for large $k$, while for small $k$ the balanced metric can still have a larger isometry group.

Based on this, we precompute the matrix symmetries with generators of $\Isom(X, g_\text{flat})$ and apply them to symmetrise the matrix after every iteration step. This helps combat the Monte Carlo integration error of the integration.

We also note that the pseudo-inverse is typically computed via singular value decomposition, and this relies on eigenvalues being reliably close to zero. This is not an issue in practice, since the output of $T^+_X$ is guaranteed to have the correct rank deficiency by construction. We also observe that our symmetrisation step does not change the eigenvalue spectrum of $H$, so it is compatible with the pseudo-inverse.

Our implementation works not just for hypersurfaces in $\PP^n$ but for general complete intersections as they are natively supported in \verb|cymyc|.

We defer the remaining details to \cref{app:implementation}.

\subsection{Dwork Family}

One particular choice of Calabi--Yaus is the Dwork family, parameterised by $\psi\in\CC$, with the defining equation
\begin{equation}
\label{eq:dwork}
    X_\psi\colon \sum_{i=1}^{n} z_i^{n} - n\psi \prod_{i=1}^{n} z_i = 0
\end{equation}
in $\PP^{n-1}$. For different choices of $n$ we have the complex torus ($n=3$), K3 surface ($n=4$), quintic threefold ($n=5$), etc.

Dwork manifolds have a large automorphism group. Based on \cite{matsumura_automorphisms_1963} we know that if $\dim_\CC X\ge 3$ and if $X$ is non-singular then the automorphism group is a finite subgroup of $\text{PGL}(n,\CC)$. For the Dwork family, this happens when $n\ge5$ and $\psi^n\ne1$. 

In the general case we have $\Aut(X_\psi) = \ZZ_{n}^{n-2} \rtimes \SS_n$. Here $\SS_{n}$ is the permutation group of the homogeneous coordinates, and $\ZZ_{n}^{n-2}$ are the toric actions $z_i \mapsto \zeta_{n}^{k_i} z_i$ where $\zeta_n$ is the $n$-th root of unity and one has $\sum k_i = 0 \pmod n$.

Following \cite{mirjanic_sharp_2026}, we conclude that $\Isom_\text{hol}(X_\psi, g_\text{CY}) = \ZZ_{n}^{n-2} \rtimes \SS_n$, while the anti-holomorphic component $\ZZ_2$ appears when $\psi\in\RR$ and manifests as complex conjugation $z\mapsto\overline z$. This enforces strong constraints on the structure of the balanced matrices $H^\PP$. Permutation invariance implies that for $\pi_1,\pi_2\in\SS_n$ we have $H\bracks{z^\alpha \overline{z^\beta}}=H\bracks{z^{\pi_1(\alpha)} \overline{z^{\pi_2(\beta)}}}$. Meanwhile, from toric invariance we get that $H\bracks{z^\alpha \overline{z^\beta}}\ne0$ iff $\alpha_i-\beta_i=\text{const}\pmod n$ for all $1\le i\le n$. On the Fermat point $X_0$ this is further restricted to $\alpha_i-\beta_i=0\pmod n$. Therefore, $H^\PP$ will be very sparse and will have many repeated entries, so we will benefit from enforcing these symmetries. Finally, when $\psi\in\RR$, the entire matrix is real by conjugation symmetry. 

Important locations in the Dwork moduli space are the Fermat point $\psi=0$ where the manifold acquires another $\ZZ_{n}$ symmetry, the conifold points at roots of unity $\psi=\zeta_{n}$ where the manifold is singular, and the large complex structure limit (LCSL) $\abs{\psi}\gg1$.

\begin{figure}[ht]
    \centering
    \includegraphics{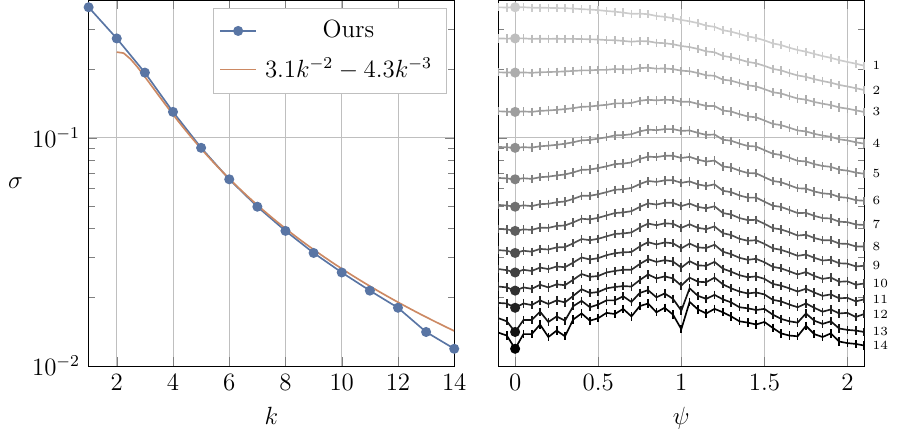}
    \caption{Left: Our $\sigma$-measures (\ref{eq:sigma-measure}) for balanced metrics on Fermat quintic $X_0$ compared to scaling law of \cite{douglas_numerical_2008} as reported in \cite{headrick_energy_2010}. Right: $\sigma$-measures on Dwork threefolds $X_\psi$ for real $\psi$ between $0$ and $2$ and $1\le k\le 14$. Darker colour corresponds to higher $k$, which yields lower $\sigma$-measure.}
    \label{fig:quintic_sigma_balanced}
\end{figure}

One way to estimate convergence without computing costly higher-order derivatives is via the $\sigma$-measure \cite{larfors_learning_2021,berglund_cymyc_2024,braun_calabi-yau_2008,ashmore_machine_2020}
\begin{equation}\label{eq:sigma-measure}
    \sigma = \int_X \abs*{1 - \frac{\det g_\text{predicted}}{\det g_\text{flat}}} \dVol_\Omega.
\end{equation}
We compute balanced metrics and their associated $\sigma$-measures for different Dwork threefolds and show dependence on $\psi$ in \cref{fig:quintic_sigma_balanced} (right). On \cref{fig:quintic_sigma_balanced} (left) we also observe that our $\sigma$-measures at Fermat point $X_0$ are in agreement with literature, empirically confirming our claim that we converge to the same potential. Note that $k=13,14$ lines have slightly larger variance due to using half as many points in order to stay within memory constraints.

\subsection{Case Study: Dwork Zerofold}

To further evaluate our implementation of Donaldson's algorithm, we seek a setting where ground truth is directly computable.

This has not been done before, and the main obstacle is that computing balanced metrics analytically requires integration over the Calabi--Yau. To circumvent this, we focus on the \emph{zerofold}
\begin{equation}
    z_1^2 + z_2^2 - 2\psi z_1 z_2 = 0.
\end{equation}
This is an exactly solvable $0$-dimensional manifold
\begin{equation}
    (z_i : z_j) \in \braces*{(1 : \psi + \sqrt{\psi^2 - 1}), (1 : \psi - \sqrt{\psi^2 - 1})} \subset \PP^1.
\end{equation}
Being a union of $2$ discrete points, it does not have a well defined metric or curvature, but Donaldson's algorithm can still be executed to find $H^\PP$. Furthermore, the integral now reduces to a sum, so one can solve it directly. Specifically, we use $k=2$, the smallest $k$ for which $s^\PP$ becomes linearly dependent on the manifold. 
\begin{prop}\label{prop:zerofold}
Using sections of degree $k=2$, the fixed point of Donaldson's algorithm for the zerofold is
\begin{align}
        H^\PP \bracks*{\abs{z_i}^4} &= \parens*{\abs{\psi}^2+\frac{1}{2}}^2 + \abs{\psi}^2 \abs{\psi^2-1}\\ 
        H^\PP \bracks*{\abs{z_i}^2 \abs{z_j}^2} &= \abs{\psi^2-1}\\ 
        H^\PP \bracks*{\abs{z_i}^2 z_i \bar{z}_j} &= \abs{\psi^2-1}\overline{\psi} \\ 
        H^\PP \bracks*{z_i^2 \bar{z}_j^2} &= -\parens*{\abs{\psi}^2+\frac{1}{2}}^2 + \abs{\psi}^2 \abs{\psi^2-1}
\end{align}
up to scaling, with other entries being zero.
\end{prop}
We confirm that our numerical predictions are within $10^{-13}$ of analytic formulae. Regarding convergence, we observe that if initial condition is Hermitian and invariant under the permutation of $z_1$ and $z_2$ then the algorithm terminates after just one iteration. 

\subsection{Numerical Balanced Metrics of the Dwork Family}
\label{sec:numerical_donaldson_dwork_small}

Having introduced a new setting in which our algorithm verifiably converges to balanced metrics, and having showed that our $\sigma$-measures on the quintic match previously known values, we now turn to the main distinguishing feature of our iteration scheme. This is the ability to compute well-defined interactions between ambient space monomials.

We begin with one-dimensional complex torus, as this is where our advantage gets displayed the soonest. For monomial degree $k=2$, \cref{fig:torus_k2} (left) shows the dependency of $H$ matrix on complex structure $\psi$ as it varies on the real axis between $0$ and $2$. This includes the Fermat point $X_0$ and the singular manifold $X_1$. Since the potential does not depend on scale of $H$, we arbitrarily normalise one of the entries to $1$. This is usually the largest entry. However, if the largest entry changes within the sampled range, we stick with a single choice for clarity.

\begin{figure}[ht]
    \centering
    \includegraphics[width=\textwidth]{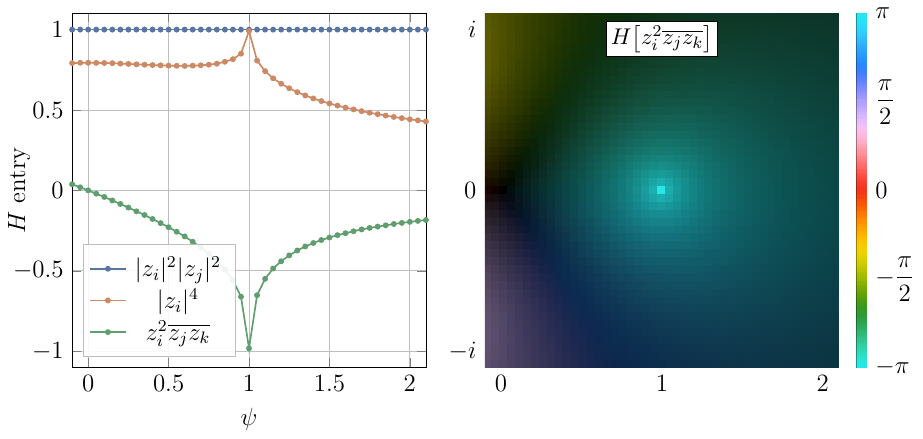}
    \caption{Left: Entries in the balanced $H$ matrix on the complex torus as a function of complex structure $\psi$. Right: Off-diagonal $H\bracks*{z_i^{2}\overline{z_j}\overline{z_k}}$ as a function of $\psi$ in the complex plane.}
    \label{fig:torus_k2}
\end{figure}

Diagonal elements $H\bracks*{|z_i|^2|z_j|^2}$ and $H\bracks*{|z_i|^4}$ are always real and non-negative, because $H$ is positive-semidefinite. On the other hand, $H\bracks*{z_i^{2}\overline{z_j}\overline{z_k}}$ can take on complex values. This is shown in \cref{fig:torus_k2} (right), where $\psi$ varies in complex plane centred on the conifold point $X_1$. Each pixel corresponds to a different Calabi--Yau, with brightness determined by magnitude, and hue by complex angle. Cyan corresponds to negative real values and red to positive reals, as depicted on the color bars to the right of each plot. Specifically, we observe that $\arg H\bracks*{z_i^{2}\overline{z_j}\overline{z_k}} \approx \arg(-\overline\psi)$, although this identification is not perfect and is slightly distorted around the conifold.

The above is a warm-up example that could have been computed with the regular Donaldson's algorithm. However, increasing section degree to $k=3$ and analysing $H$ does require using our approach. \cref{fig:torus_k3} (top left) shows that the entries that we compute behave similarly to the previous $k=2$ case. This is in fact quite desirable, as it suggests $H$ has generalisable patterns one might be able to discover.

\begin{figure}[ht]
    \centering
    \includegraphics[width=\textwidth]{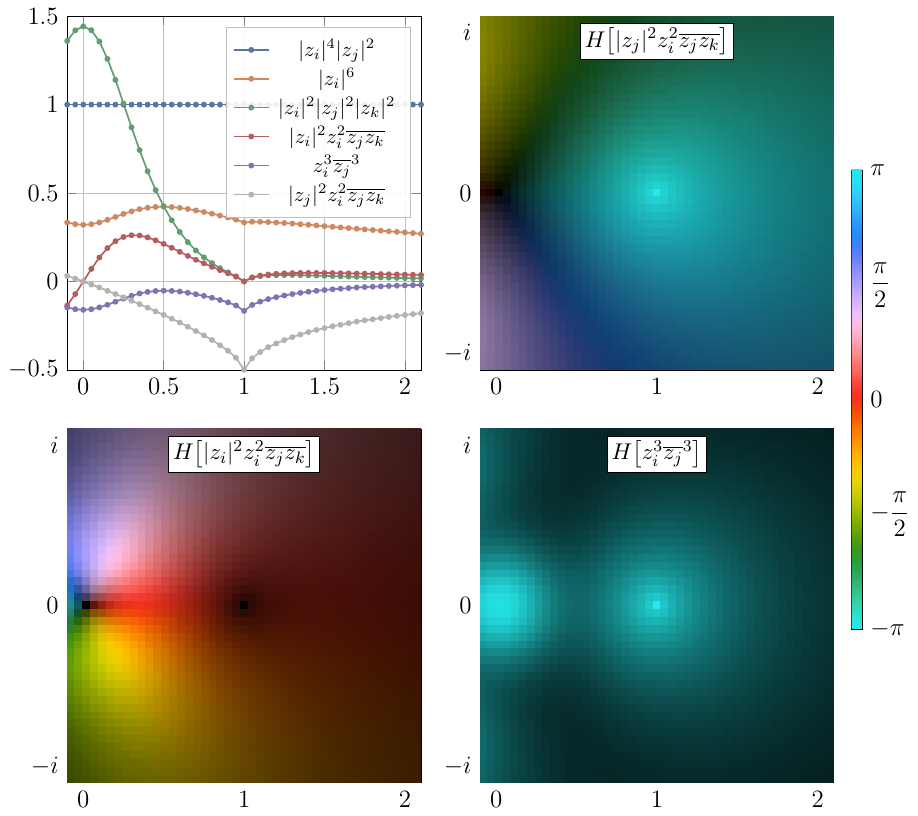}
    \caption{Entries in the $H$ matrix for the torus with degree $k=3$ in the ambient basis.}
    \label{fig:torus_k3}
\end{figure}

The rest of \cref{fig:torus_k3} shows the three off-diagonal elements as $\psi$ varies in the complex plane. Curiously, $H\bracks*{|z_j|^{2}z_i^{2}\overline{z_j}\overline{z_k}}$ behaves similarly to $H\bracks*{z_i^{2}\overline{z_j}\overline{z_k}}$, while $H\bracks*{|z_i|^{2}z_i^{2}\overline{z_j}\overline{z_k}}$ appears completely different. Also, even though $H\bracks*{z_i^{3}\overline{z_j}^{3}}$ is forced to be real by symmetry, there is no \emph{a priori} reason for it to be negative for all $\psi\in\CC$.

At this point we can move beyond qualitative descriptions and attempt to write down precise statements about the $H$ matrix. For example, consider the level $2$ $H$ matrix on the Fermat cubic $X_0$. Symmetry forces $H\bracks*{z_i^2\overline{z_jz_k}}=0$, but does not fix $H\bracks*{\abs{z_i}^4}$. Nevertheless, numerically we find $H\bracks*{\abs{z_i}^4}\big\vert_{\psi=0}\approx 0.7938 \approx 2^{-1/3}$.
\begin{conj}
    On the Fermat cubic $X_0$, the first few balanced potentials are
    \begin{align}
        K^{(2)} &= \frac{1}{2\pi}\log\parens*{2^{-1/3}\sum\abs{z_i}^4 + \sum \abs{z_i}^2\abs{z_j}^2}\\
        K^{(3)} &= \frac{1}{3\pi}\log\parens[\Big]{3^{-1/3}\sum \abs{z_i}^4\abs{z_j}^2 + \sum \abs{z_i}^2\abs{z_j}^2\abs{z_k}^2 + \frac{2}{9}\sum \abs{z_i}^6 - \frac{1}{9}\sum z_i^3\overline{z_j}^3}\\
        K^{(4)} &= \frac{1}{4\pi}\log\!\Big( 2^{-1/3} \sum \abs{z_i}^4\abs{z_j}^4 + 2^{1/3} \sum \abs{z_i}^4\abs{z_j}^2\abs{z_k}^2 + \frac{1}{3}\sum\abs{z_i}^6\abs{z_j}^2 \nonumber\\
        &+\frac{2}{9}\sum\abs{z_i}^8 - \frac{2}{9}\sum\abs{z_j}^2 z_i^3\overline{z_k}^{3} - \frac{1}{9}\sum\abs{z_j}^2 z_i^3\overline{z_j}^{3} - \frac{1}{9}\abs{z_i}^2 z_i^3\overline{z_j}^{3} \Big).
    \end{align}
\end{conj}
While numerically supported, it is unclear how one would prove this. Furthermore, the existence of algebraic numbers and loss of precision for higher $k$ make it unclear how to generalise these potentials to higher dimensions. A better point in the moduli space turns out to be the conifold $X_1$.
\begin{conj}
    On the singular $X_1$ onefold, the entries of the $H$ matrix exhibit cusps. Furthermore, the balanced metrics for different $k$ are
    \begin{align}
        K^{(2)} &= \frac{1}{2\pi}\log\sum \abs{z_i^2 - z_j z_k}^2 = \frac{1}{2\pi}\log\norm{\nabla Q}^2, \text{ and} \\
        K^{(3)} &= \frac{1}{3\pi}\log\parens*{\parens*{\sum\abs{z_i}^2}^3 - \Re \parens*{\parens*{z_1\overline{z_2} + z_2\overline{z_3} + z_3\overline{z_1}}^3} },
    \end{align}
    and $\nabla Q$ is the gradient of the defining polynomial from \cref{eq:dwork}.
\end{conj}
Unlike the $X_0$ case, at $\psi=1$ all $H$ entries are simple fractions which makes their interpretation easier. At the same time, $\psi=1$ is the point where the defining equation factorises into 
\begin{equation}
    \parens{z_1 + z_2 + z_3}\parens{z_1 + \zeta_3 z_2 + \zeta_3^2  z_3}\parens{z_1 + \zeta_3^2 z_2 + \zeta_3 z_3}=0.
\end{equation}
It is likely that this geometric simplification is direct cause for previously observed algebraic simplification. Beyond $\psi=0$ and $\psi=1$, though, $H$ matrix does not possess immediately discernible features at other nearby points. Therefore, we now move to the Dwork threefolds.

\begin{figure}[ht]
    \centering
    \includegraphics{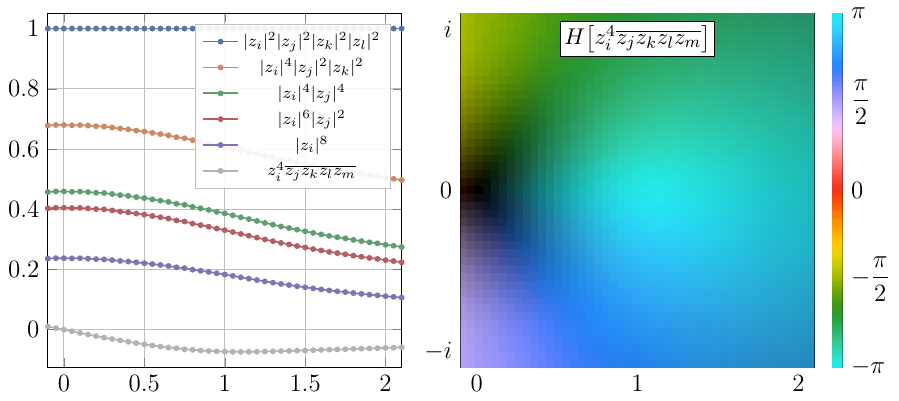}
    \caption{Entries in the $H$ matrix for the quintic with degree $k=4$.}
    \label{fig:quintic_k4}
\end{figure}

\begin{figure}[ht]
    \centering
    \includegraphics{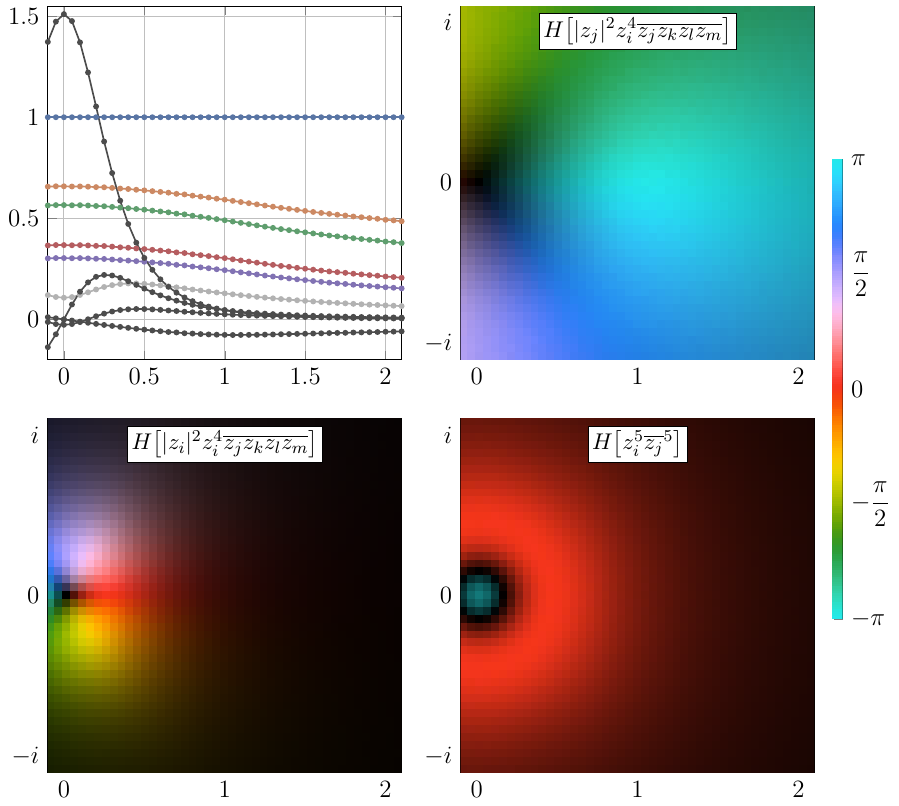}
    \caption{Entries in the $H$ matrix for the quintic with degree $k=5$ in the ambient basis.}
    \label{fig:quintic_k5}
\end{figure}

To make comparison more appropriate, we compute quintic metrics at $k=4$ (\cref{fig:quintic_k4}) and $k=5$ (\cref{fig:quintic_k5}) as they more naturally correspond to $k=2$ and $k=3$ metrics on the onefold. The off-diagonal entries $H\bracks*{z_i^{4}\overline{z_j}\overline{z_k}\overline{z_l}\overline{z_m}}$, $H\bracks*{\abs{z_j}^2 z_i^{4}\overline{z_j}\overline{z_k}\overline{z_l}\overline{z_m}}$, and $H\bracks*{\abs{z_i}^2 z_i^{4}\overline{z_j}\overline{z_k}\overline{z_l}\overline{z_m}}$ all behave like their onefold counterparts. Meanwhile, $H\bracks*{z_i^{5}\overline{z_j}^{5}}$ is now positive almost everywhere. 

Surprisingly, all of these examples are either purely real, or have complex angles close to that of $\pm\overline\psi$. For comparison, the only complex-valued entry on the zerofold has a $\overline\psi$ term (\cref{prop:zerofold}).

The most significant difference is that impact of the conifold point on $H$ is now heavily suppressed. The cusps that were characteristic of the onefolds are now either gone, or are small enough not to be detected by the sampling resolution. As a result, the entire $H$ matrix becomes approximately phase invariant in $\psi$. This near-symmetry can explain why previous experiments \cite{lee_approximate_2025} that computed approximate potentials of Ricci-flat metrics only detected dependence on $\abs{\psi}$. 

%% file: 4_Balanced_Dwork.tex
\section{Large Complex Structure Limit}
\label{sec:dwork_power_laws}

In the previous section we saw our Donaldson's algorithm in action. We also observed that the $H$ matrix becomes interpretable at highly symmetric or singular points. Now, we will explore the most singular and symmetric point in the moduli space. This is the $\abs{\psi}\to\infty$ limit, called the Large Complex Structure Limit (LCSL).

At $\psi=\infty$, the defining \cref{eq:dwork} collapses to
\begin{equation}
    X_\infty\colon \prod_{i=1}^{n} z_i = 0,
\end{equation}
and the geometry simplifies to a highly singular gluing of lower-dimensional projective spaces. It is more geometrically interesting, though, to avoid taking this na\"ive limit and resolve the singularities instead. To do this, we consider how $X_\psi$ transforms under the map $\Log(z_1, z_2, \dots) = \parens*{\log\abs{z_1}, \log\abs{z_2}, \dots}$. The image $\Log(X_\psi)$ is called an \emph{amoeba}. For large but finite $\psi$ this amoeba \emph{tropicalises} into
\begin{equation}
    \trop(X_\psi)\colon \max\parens*{n u_1, n u_2, \dots, t + \sum u_i} \text{ is saturated},
\end{equation}
where $u_i = \log |z_i|$ and $t = \log |n \psi|$. This tropical surface is called the \emph{spine} of the amoeba. The locus it defines is the set of points where $\max$ is attained by at least two terms (saturated). 

The purpose of $\trop(X_\psi)$ is to approximate $X_\psi$ near the LCSL while avoiding the algebraic collapse. Formally, $\trop(X_\psi)$ is the \emph{skeleton} of the degenerating family $X_\psi$, a concept made precise by the Gross--Siebert program \cite{gross_real_2011}. Practically, $\trop(X_\psi)$ is a real piecewise-linear manifold, which makes its study easier.

It is well established that in this limit the manifold acquires a special Lagrangian (sLag) fibration
\begin{equation}
    T^d \to X_\psi \xrightarrow{\pi} B,
\end{equation}
where $d=n-2=\dim_\CC X$, and the base $B$ is a real manifold of dimension $d$. This base is the boundary of some polytope and thus homeomorphic to $S^d$. In particular, for the Dwork family  we have $B=\partial\Sigma_{n-1}$, where $\Sigma_n=\braces*{\xi\in\RR_{\ge0}^{n+1}\colon \sum\xi_i=1}$ is the $n$-simplex. Thus, $B$ is a gluing of $n$ copies of $\Sigma_{n-2}$ that correspond to $\abs{z_i}=0$ for some $z_i$. Furthermore, the projection $\pi\colon X_\psi \to B$ is precisely the $\Log$ map.

The sLag fibration plays a major role in mirror symmetry via the SYZ conjecture of \cite{strominger_mirror_1996} that explores a duality between the Dwork manifold $X$ and its mirror $X^\vee$. 

While aspects of the original questions surrounding mirror symmetry and T-duality still remain unproven, much about the metric collapse is known today. For example, the \textquote{metric SYZ} of Refs.~\cite{gross_large_2000} and \cite{kontsevich_homological_2001} show how the Ricci-flat metric deforms as the Calabi--Yau manifolds approach LCSL. A recent landmark result by \cite{li_stromingeryauzaslow_2022} proves that the Ricci-flat metric undergoes Gromov-Hausdorff collapse onto the base of the fibration at LCSL:
\begin{equation}
    \parens*{X_\psi, g_\mathrm{CY}} \xrightarrow[\psi\to\infty]{\mathrm{GH}} \parens*{B, g_\mathrm{MA}}.
\end{equation}
To study the Ricci-flat metric immediately before this collapse, we introduce the \emph{semi-flat metric}. Except on a locus $\Delta\subset B$ of codimension 2, the $T^d$ fibers collapse smoothly and we can define logarithmic coordinates $(u, \theta)$, where $u=\Log(z)$ are affine coordinates on $B$ and $\theta_i=\arg z_i$ are angular coordinates on the fibers. Then, the semi-flat metric is defined as
\begin{equation}
g_{\text{SF}} =\sum_{i,j=1}^{d} \parens*{\Hess\phi}_{ij} \parens{  \underbrace{\dd u_i \otimes \dd u_j}_{g_\text{base}} +  \underbrace{\dd \theta_i \otimes \dd \theta_j}_{g_\text{fiber}} },
\end{equation}
where $\phi(u)$ is a locally strictly convex potential function on the base. It is called semi-flat because it is flat on the $T^d$ fibers by construction. For some choice of $\phi$, $g_{\mathrm{SF}}$ will be Ricci-flat over $B\setminus\Delta$. This happens when $\phi$ is solution to a \emph{real} Monge--Amp\`ere equation
\begin{equation}\label{eq:monge-ampere-real}
    \det \Hess \phi = \text{const},
\end{equation}
which is itself a degeneration of the complex Monge--Amp\`ere (MA) equation
\begin{equation}\label{eq:monge-ampere}
    \det g_\text{CY} = \kappa \Omega \overline \Omega
\end{equation}
that the Ricci-flat metric satisfies. Here $\Omega$ is the holomorphic volume form, and $\kappa\in\RR$ depends only on $\psi$ and fixes the K\"ahler class to match the ambient metric.

More traditionally, we can adopt the symplectic view and use the action-angle coordinates $(p,\theta)$, where $p_i = \partial \phi/\partial u_i$, and write the semi-flat metric as 
\begin{equation}
\underbrace{\sum_{i,j=1}^{d} \parens*{\Hess\phi^*}_{ij} \dd p_i \otimes \dd p_j}_{g_\text{base}} + \underbrace{\sum_{i,j=1}^{d} \parens*{\Hess\phi^*}^{-1}_{ij} \dd \theta_i \otimes \dd \theta_j}_{g_\text{fiber}},
\end{equation}
Note the appearance of matrix inverse for the fibers, which results in this metric having a very simple symplectic form
\begin{equation}
    \omega = \sum_{i=1}^n \dd p_i \wedge \dd \theta_i.
\end{equation}
Furthermore, note that in these coordinates we use $\phi^*(p):=\max(p\cdot u - \phi(u))$, which is the Legendre conjugate of $\phi$. This duality is one of the hallmarks of SYZ \cite{hitchin_moduli_1997}.
\begin{figure}[ht]
    \centering
    \includegraphics{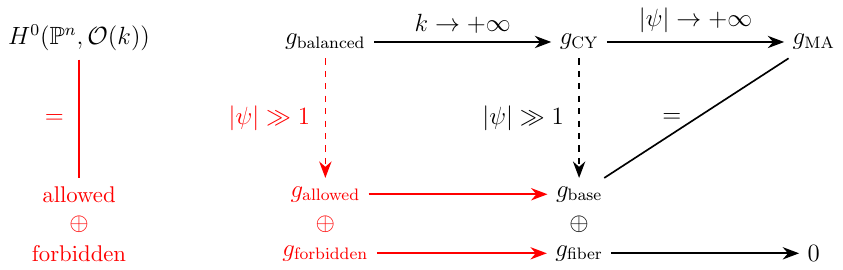}
    \caption{Near the large complex structure limit, Calabi--Yau acquires special Lagrangian fibration and the Ricci-flat metric collapses onto the base. We propose a similar behaviour for the balanced metrics that is governed by the global sections of the ambient projective space.}
    \label{fig:metric_collapse}
\end{figure}

In the LCSL the fibers collapse, which causes the $\dd \theta_i \otimes \dd \theta_j$ terms to vanish. What remains is the real Monge--Amp\`ere metric on the base, and this is exactly equal to $g_\text{MA}$. In this way, $g_\text{SF}$ approximates the true Ricci-flat metric $g_\text{CY}$. However, $g_\text{SF}$ is itself not Ricci-flat on the entire manifold, since its curvature explodes at $\Delta$. Near $\Delta$, the true metric $g_{\mathrm{CY}}$ deviates from $g_\text{SF}$ to smooth out the curvature and is locally described by quantum corrections of \cite{ooguri_summing_1996}. As $\psi \to \infty$, these corrections become exponentially suppressed, yielding $g_\text{SF}\approx g_\text{CY}\to g_\text{MA}$. The full picture is depicted in \cref{fig:metric_collapse}, along with our contributions from the balanced metrics that we will discuss below.

\subsection{Asymptotics of Balanced Potentials at LCSL}\label{sec:asymptotic_potential}

We compute the balanced metrics for $\abs{\psi}\to\infty$ and plot them in \cref{fig:donaldson_coefficients_inf}, with absolute values of different $H^\PP$ entries on y-axes and $\psi\in\bracks{1,10^3}$ on x-axes. While we only use real $\psi$ for this plot, we observe the power laws really are dependent on $\abs{\psi}$ for $\psi\in\CC$. This is consistent with the complex plane figures from \cref{sec:numerical_donaldson_dwork_small}.

\begin{figure}[htp]
    \centering
    \includegraphics[width=\linewidth]{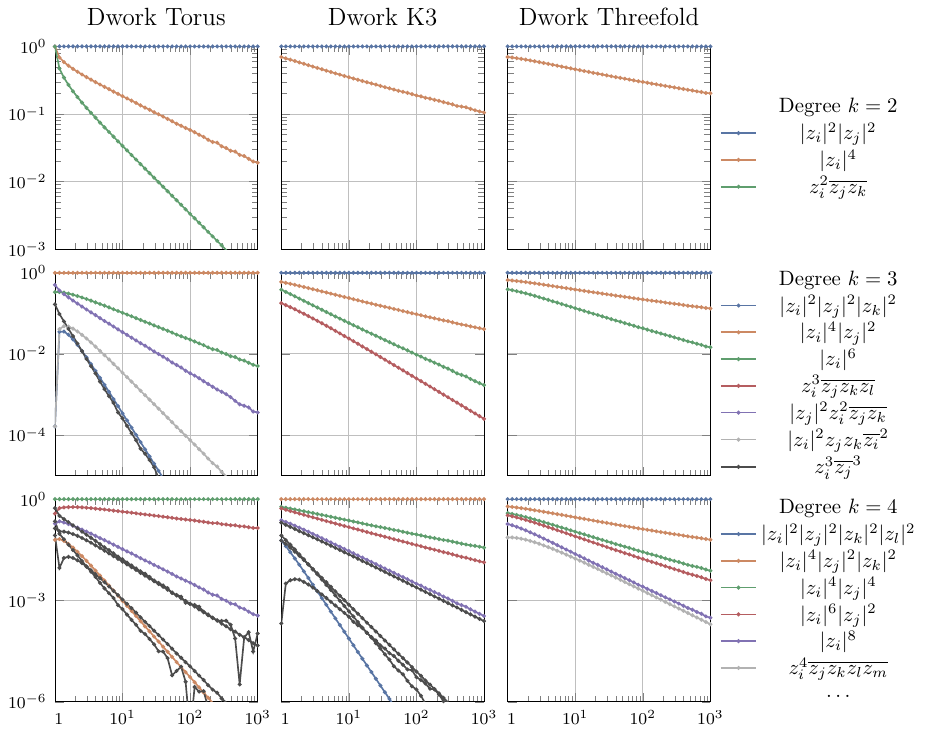}
    \caption{Absolute values of $H$ entries (y-axis) versus $\psi$ (x-axis). Columns correspond to torus, K3, and quintic. Rows correspond to section degrees $k$. Near the LCSL, all $H$ entries decay as power laws.}
    \label{fig:donaldson_coefficients_inf}
\end{figure}

We render the balanced metric data in log-scale, and observe linear trends that correspond to power laws $y\sim a x^b$ with $b\le0$. Significantly, we observe that all entries in $H$ are decaying as power laws, even as we vary degree $k$ (vertically), dimension (horizontally), or types of sections inside each plot. We always normalise $H$ so that $\max_{ij} \abs{H_{ij}} = 1$ in order to eliminate the scale degree of freedom from the matrix. Therefore, power laws that we discuss are more properly the power laws of $\abs{H_{ij}} / \max_{ij} \abs{H_{ij}}$.

These power laws are distorted by the conifold point $\psi=1$, as expected. Otherwise, we observe they have little noise even for very large $\psi$. The numerical noise only impacts sections when they drop several orders of magnitude below the dominant one.

Since these power laws are very well behaved, we can proceed to numerically estimate them from this data, as discussed in \cref{app:implementation}. Specifically, we extract power laws via linear regression on 
\begin{equation}
    \log \abs{H_{ij}} = a + b \log\abs{n\psi}.
\end{equation}
To interpret these results, we split them based on whether they are compatible with the geometry of $X_\infty$ or not. We note that $X_\infty$ has more isometries than $X_\psi$ for finite $\psi$ and this imposes extra constraints on some sections. For example, toric actions generalise to $z_k \mapsto e^{i \theta_k} z_k$ for arbitrary phases $\theta$, which forces $H^\PP$ to be diagonal in the limit. Similarly, sections that contain all ambient coordinates will vanish on $X_\infty$, so they cannot contribute to the metric. Therefore, we say that the entry $H\bracks{z^\alpha \overline{z^\beta}}$ is \emph{allowed} if the vectors $\alpha$ and $\beta$ are equal (so the entry is diagonal) and if there is some index $i$ for which $\alpha_i=\beta_i=0$ (so $z^\alpha \overline{z^\beta}$ does not vanish). Otherwise, the entry is \emph{forbidden}.

By the above argument we conclude that
\begin{equation}
    \lim_{\psi\to\infty} H[z^\alpha \overline{z^\beta}]=0
\end{equation}
if $z^\alpha \overline{z^\beta}$ is forbidden. However, the data is more informative and shows that all entries except one approach zero. The following conjectures formalise the observed behaviour.
\begin{conj}
    For allowed entries, the decay rate of $H\bracks{z^\alpha \overline{z^\alpha}}$ is correlated with variance of non-zero entries of $\alpha$. Specifically, the unique $z^\alpha$ with only one zero in the multidegree and with all other degrees differing by at most one will be dominant in the limit.
\end{conj}
For example, we expect $H\bracks{ \abs{z_1}^{8} } \ll H\bracks{ \abs{z_1}^{2} \abs{z_2}^{6} } \ll H\bracks{ \abs{z_1}^{4} \abs{z_2}^{4} } \ll H\bracks{ \abs{z_1}^{4} \abs{z_2}^{2} \abs{z_3}^{2} } $ for K3 and the quintic. Not for the complex torus, however, since on it $H\bracks{ \abs{z_1}^{4} \abs{z_2}^{2} \abs{z_3}^{2} }$ is forbidden because it contains all three coordinates. This behaviour is in fact visible in the last row of \cref{fig:donaldson_coefficients_inf}. 
\begin{conj}
    For the forbidden entries, they decay at least as fast as 
    \begin{equation}
        \mathcal{O}\parens*{\abs{\psi}^{-\min\parens{\alpha_i + \beta_i}}},
    \end{equation}
    or $\mathcal{O}(\abs{\psi}^{-1})$, whichever is stronger.
\end{conj}
\begin{coro}
    Only the allowed entries can decay slower than $\mathcal{O}(\abs{\psi}^{-1})$.
\end{coro}
We observe that allowed entries can decay arbitrarily slowly, and that the multiplicative constant in front of $\log\abs{n\psi}^b$ is always equal to $1$.

We now turn to the asymptotic behaviour of balanced metrics at LCSL.
Since there is only one dominant term $z^\alpha\overline{z^\alpha}$ in $H$ up to permutation, one might consider approximating the balanced potential with
\begin{equation}
    \frac{1}{k}\log\parens*{a \abs{n\psi}^{b} z^\alpha \overline{z^\alpha}} = \frac{1}{k} \sum \alpha_i \parens*{\log z_i + \log \overline{z_i}} + \text{const}.
\end{equation}
In this approximation, all terms are either purely holomorphic or anti-holomorphic. This will not 
suffice, since
taking $\partial\overline\partial$ would imply that the metric is identically zero. 

Therefore, we must distinguish order of limits near the LCSL. While evaluating
\begin{equation}
    \lim_{k\to\infty}\lim_{\psi\to\infty}K_\text{balanced}
\end{equation}
produces the degenerate metric described above,
\begin{equation}
    \lim_{\psi\to\infty}\lim_{k\to\infty}K_\text{balanced}
\end{equation}
will preserve geometric information. Of course, in numerical computations both $k$ and $\psi$ are large but finite, and so the limit behaviour we extract from them are conjectured extrapolations that cannot replace a proper proof.

With this in mind, a better way to capture the limit behaviour is to approximate balanced potential by discarding forbidden terms:
\begin{equation}\label{eq:k_approx_def}
    K_\text{allowed} := \frac{1}{k}\log\parens*{\sum_{\overset{\text{allowed } \alpha}{\abs{\alpha}=k}} \abs{n\psi}^{-f(\alpha)} \abs*{z^\alpha}^2}.
\end{equation}
According to our conjectures, this is indeed an accurate approximation to the true potential. Crucially, since $K_\text{allowed}$ contains multiple terms inside $\log$, its complex Hessian is non-trivial resulting in a non-vanishing metric.

Next, we rewrite \cref{eq:k_approx_def} as 
\begin{equation}
    K_\text{allowed} = \frac{1}{k}\log\sum\exp\parens*{-f(\alpha)\log\abs{n\psi} + 2 \sum \alpha_i \log\abs{z_i}}
\end{equation}
Let $\xi_i = \alpha_i / k$ be the fractional multidegrees. Recall that $u_i = \log\abs{z_i}$, and let $t=\log\abs{n\psi}$. Finally, let $F(\vec \xi) = \lim_{k\to\infty} f(k \vec \xi) / k$. Using the well known limit 
\begin{equation}
    \lim_{k\to\infty} \frac{1}{k}\log\sum e^{k x_i}= \max x_i,
\end{equation}
we obtain a tropical limit
\begin{equation}
    K_\text{trop} := \lim_{k\to\infty}K_\text{allowed} = \max_{\overset{\text{allowed }\xi}{\sum \xi_i=1}} \parens*{2 u \cdot \xi - t F(\xi) } = t F^*\parens*{\frac{2u}{t}}.
\end{equation}
This potential is geometrically meaningful. Since the allowed monomials are precisely the ones that are compatible with the $\Log$-map projection to the base $B$, we conclude that $K_\text{trop}$ is a strong candidate for the semi-flat metric potential. The precise identification requires checking the real Monge--Amp\`ere \cref{eq:monge-ampere-real}, and we perform this numerically in the following sections.

At the same time, $K_\text{trop}$ is (up to appropriate scaling) equal to the Legendre transform of $F(\xi)$, which we will denote with $F^*(x)$. It is well known that Legendre transform of $f$ satisfies
\begin{equation}
    \parens*{\det \Hess f}\parens*{\det \Hess f^*} = 1.
\end{equation}
\begin{coro}
    $F^*$ satisfies the MA equation $\det\Hess F^*=\textit{const.}$ iff $F$ satisfies the same equation.
\end{coro}
We have thus arrived at a relationship between the tropical potential and the exponents of sections that this potential is made of. This duality is understood to correspond to mirror symmetry. Since $K_\text{trop}$ is a K\"ahler potential on $X_\psi$, $F$ is its dual potential on the mirror $X^\vee$. 

\begin{figure}[ht]
    \centering
    \includegraphics{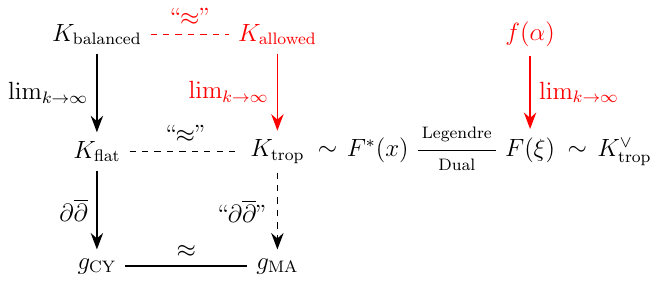}
    \caption{Tropical approximation to the flat metric potential is Legendre dual to, and thus completely determined by, exponent decay rates $f$. Dashed lines indicate conjectural identification of $K_\text{trop}$ with the real Monge--Amp\`ere metric potential on $B$. Red components are our finite approximations based on balanced metrics.}
    \label{fig:trop_potential_duality}
\end{figure}

To see how the metric associated with $K_\text{trop}$ scales with $t=\log\abs{n\psi}$, we can compute 
\begin{equation}
    (\Hess K_\text{trop})(u) = \frac{4}{t} (\Hess F^*)\parens*{\frac{2u}{t}}
\end{equation}
Since $F$ and $F^*$ have no dependency on $t=\log\abs{n\psi}$ by construction, we conclude that $\Hess K_\text{trop}\sim1/\log\abs{\psi}$. Now we can compute the diameter of the fibers as 
\begin{equation}
    D_\text{fiber} = \int_0^{2\pi} \sqrt{g_{\theta\theta}} \dd \theta \sim \frac{1}{\sqrt{\log\abs{\psi}}}.
\end{equation}
Similarly, the base satisfies
\begin{equation}
    D_\text{base} \sim \int_B \frac{1}{\sqrt{\log\abs{\psi}}} \dd u.
\end{equation}
Since $u$ scales with $\log\abs{\psi}$ we have $D_\text{base}\sim\sqrt{\log\abs{\psi}}$. Therefore, the total volume remains bounded. This perfectly captures our previous choice to compute balanced metrics in the regime where total volume is fixed.

Furthermore, we can compute that sLag fibers collapse at the rate
\begin{equation}
    \frac{D_\text{fiber}}{D_\text{base}}\sim \frac{1}{\log\abs{\psi}}.
\end{equation}
Compare this to \cite{gross_large_2000}, in which elliptic curve $E$ with fixed volume is identified with $\CC/\innerprod{1/\sqrt\alpha, i\sqrt\alpha}$, where $\alpha\to\infty$. The first dimension corresponds to fibers that shrink like $1/\sqrt\alpha$, while the base expands at the rate $\sqrt\alpha$, just like above.

\subsection{Case Study: Dwork Torus}
\label{sec:dwork_torus_sections}

To understand the behaviour of $K_\text{allowed}$ and $K_\text{trop}$, we consider the complex torus as an example before moving to higher-dimensional manifolds. As per our conjectures in \cref{sec:asymptotic_potential}, the monomials that compose $K_\text{allowed}$ are $\abs{z_i}^{2\alpha_i}\abs{z_j}^{2\alpha_j}$, and so $\xi = \frac{\alpha_i}{\alpha_i + \alpha_j}$.
\begin{prop}
    The limit coefficients $F$ on the torus are 
    \begin{equation}
        F(\xi) = \parens*{\xi - \frac{1}{2}}^2\times\mathrm{const}
    \end{equation}
    \begin{proof}
    Since $F$ is univariate we have $\Hess F = F''$. Therefore, the Monge--Amp\`ere \cref{eq:monge-ampere-real} simplifies to $F(\xi) = a\xi^2 + b\xi + c$ for some $a$, $b$, $c$. 
    
    Next, $F$ must be invariant under $\alpha_i\leftrightarrow\alpha_j$ interchange, because this preserves the type of monomial. Expressing this in terms of $\xi$ we obtain $F(\xi) = F(1-\xi)$.
    
    Finally, the dominant monomials have $\alpha_i=\alpha_j$, which corresponds to $\xi=1/2$, so we have $F(1/2)=0$. The formula for $F$ immediately follows.
    \end{proof}
\end{prop}
Translating the above result to exponents $f$ at finite $k$, we conclude that exponents in power laws for the complex torus behave like
\begin{equation}
    f(\alpha_i, \alpha_j) \sim \frac{(\alpha_i - \alpha_j)^2}{\alpha_i + \alpha_j}
\end{equation}
when $k$ tends to infinity. Since the Monge--Amp\`ere \cref{eq:monge-ampere-real} only holds in the limit, we cannot conclude more about the finite $k$ behaviour using this approach. Nevertheless, we can observe it from data. 
\begin{conj}\label{conj:torus_decay}
    On the complex torus, the power law decay at LCSL is given by
    \begin{align}
        f(\alpha_i, \alpha_j) &= \frac{\floor*{(\alpha_i-\alpha_j)^2/4}}{\alpha_i + \alpha_j}\\
        H^\PP\bracks*{\abs{z_i}^{2\alpha_i}\abs{z_j}^{2\alpha_j}} &= \abs{3\psi}^{-\frac{\floor*{(\alpha_i - \alpha_j)^2/4}}{\alpha_i + \alpha_j}} + \text{lower order terms}
    \end{align}
\end{conj}
This conjecture is corroborated by data in \cref{fig:torus_allowed_sections}. On the left, we plot all exponents for different values of $k$, and observe that they match the predictions for small and large $k$ alike. Floor function $\floor*{\cdot}$ is in the formula in order to capture the small $k$ behaviour. Quadratic limit is not visible in the left plot, so we transform data to bring out the convergence to $F$ on the right.

\begin{figure}[ht]
    \centering
    \includegraphics{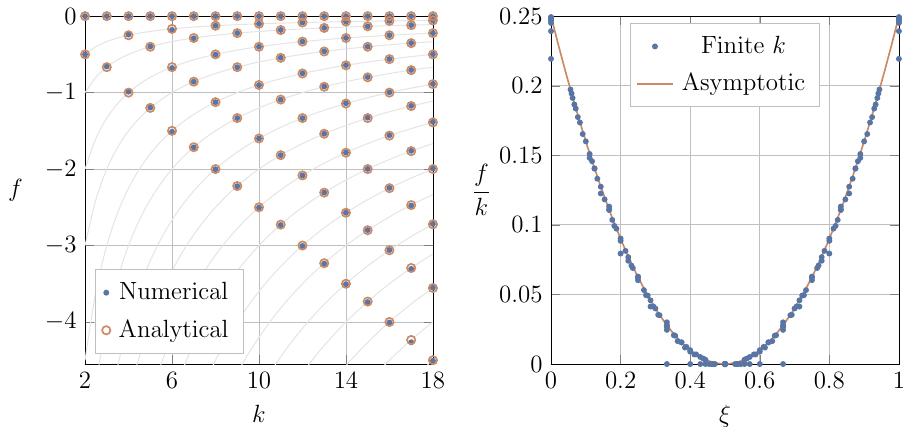}
    \caption{Left: Numerically obtained exponents match conjecture. Thin hyperbolic curves $y=\floor{n^2/4}/x$ for different integers $n$ pass through all points. Right: In an appropriate coordinate frame, the exponents converge to $\parens{\xi-1/2}^2$.}
    \label{fig:torus_allowed_sections}
\end{figure}

Consider a family of monomials where $\alpha_i - \alpha_j$ is kept fixed while $k=\alpha_i+\alpha_j$ increases. These monomials will decay with exponent $f\sim \text{const}/k$. That is, they will decay arbitrarily slowly as $k$ increases, and will therefore influence the $k\to\infty$ limit even though they vanish if $\psi\to\infty$ limit is taken first.

As a concrete example of the above, leading monomials for even $k$ are $\abs{z_i}^{2k} \abs{z_j}^{2k}$. Subleading monomials are $\abs{z_i}^{2(k+1)} \abs{z_j}^{2(k-1)}$, and thus we have $f(k-1,k+1)=1/2k$.
 
Applying \cref{conj:torus_decay} to $K_\text{allowed}$ we obtain
\begin{equation}
    K_\text{allowed} = \frac{1}{k}\log\parens*{\sum_{i\ne j}\sum_{\alpha_i+\alpha_j=k} \parens*{3\abs{\psi}}^{-\frac{\floor*{(\alpha_i-\alpha_j)^2/4}}{\alpha_i+\alpha_j}} \abs{z_i}^{2\alpha_i}\abs{z_j}^{2\alpha_j}}
\end{equation}
\begin{prop}
    The tropical limit of $K_\mathrm{allowed}$ is
    \begin{equation}
        K_\mathrm{trop} = \max_{i\ne j}\parens*{ \frac{\parens*{\log\abs{z_i}-\log\abs{z_j}}^2}{\log\abs{3\psi}} + \log\abs{z_i} + \log\abs{z_j}}
    \end{equation}
    \begin{proof}
        As usual, let $u_i = \log\abs{z_i}$, $t=\log\abs{3\psi}$, $\xi_i=\alpha_i/k$. For large $k$ we can ignore $\floor{\cdot}$. Now, we can rewrite $K_\text{allowed}$ to
        \begin{equation}
            K_\text{allowed} = \frac{1}{k} \log \sum \exp \parens*{-\frac{kt}{4}(\xi_i-\xi_j)^2 + 2k\xi_i u_i + 2k\xi_j u_j}
        \end{equation}
        Famously, $\lim_{k\to+\infty} \frac{1}{k}\log\sum\exp(k x_i)=\max x_i$ so we have
        \begin{equation}
            \lim_{k\to\infty}K_{approx} = \max_{i\ne j, \xi_i, \xi_j} \parens*{-\frac{t}{4}(\xi_i-\xi_j)^2 + 2\xi_i u_i + 2\xi_j u_j}
        \end{equation}
        To eliminate $\xi_i$ and $\xi_j$ we note $\xi_i+\xi_j=1$ and set $\delta=\xi_i-\xi_j$, $\xi_i=\frac{1+\delta}{2}$, $\xi_j=\frac{1-\delta}{2}$, so
        \begin{equation}
            \lim_{k\to\infty}K_{approx} = \max_{i\ne j, \delta} \parens*{-\frac{t}{4}\delta^2 + (u_i - u_j)\delta + (u_i + u_j)}
        \end{equation}
        This is quadratic in $\delta$ so we can easily find the maximum at $\delta = 2 (u_i - u_j) / t$, yielding
        \begin{equation}
            \lim_{k\to\infty}K_{approx} = \max_{i\ne j} \parens*{u_i + u_j + (u_i - u_j)^2/t}
        \end{equation}
        as desired.
    \end{proof}
\end{prop}
An important point in the above derivation is that even though $\max$ suggests piecewise linear behaviour, the resulting potential is convex with a non-vanishing metric. This happens because the number of monomials increases with $k$ and so the regions that would be piecewise linear shrink in size and geometry approaches a smooth surface.
\begin{prop}
    $K_{trop}$ is flat on the base.
    \begin{proof}
        Suppose without loss of generality that we are on the patch where $\abs{z_1} \approx \abs{z_3} \gg \abs{z_2}$ and let $w = \log \abs{z_1} - \log \abs{z_2}$. After a K\"ahler transformation that eliminates $\log\abs{z_1} + \log\abs{z_2}$, the potential becomes $K=w^2/t$. Then, we find that $g_{w\bar w} = 2/t$ is constant in $w$ and hence satisfies the real MA equation.
    \end{proof}
\end{prop}
This shows that $K_\text{trop}$ is the K\"ahler potential corresponding to the semi-flat metric of the complex torus, and while the true Ricci-flat metric remains unknown, we can use $K_\text{trop}$ in computations. For example, we can use it to obtain the real constant $\kappa$ in the complex Monge--Amp\`ere \cref{eq:monge-ampere}.
\begin{prop}
    On the 1-dimensional Dwork torus, the Monge--Amp\`ere \cref{eq:monge-ampere} constant $\kappa$ behaves asymptotically like
    \begin{equation}
        \kappa \sim \frac{\abs{\psi}^2}{2\log{\abs{3\psi}}}
    \end{equation}
    at the LCSL.
    \begin{proof}
    Suppose without loss of generality that we are on a patch where $\abs{z_2}\approx \abs{z_3} \gg \abs{z_1}$. For hypersurfaces, it is customary to fix an affine patch where e.g.~$z_3=1$ and the local coordinate is e.g.~$z_1$. With this, the semi-flat potential becomes $K_\text{trop}=(\log \abs{z_1})^2 / \log\abs{3\psi}$, and 
    \begin{equation}
        g_{z_1\overline z_1} = \frac{1}{2\log\abs{3\psi}} \frac{1}{\abs{z_1}^2}
    \end{equation}
    On the other hand, one can compute the volume form and show that $\Omega = \psi^{-1} \text{d}z_1/z_1$. Substituting these into the Monge--Amp\`ere equation $g = \kappa \Omega\overline\Omega$, we obtain $\kappa$.
    \end{proof}
\end{prop}
At the same time, $\kappa$ can be numerically estimated by integrating the volume elements over the manifold. \cref{fig:kappa_scaling} shows that the numerical and analytical values for $\kappa$ are perfectly aligned.

\begin{figure}
    \centering
    \includegraphics{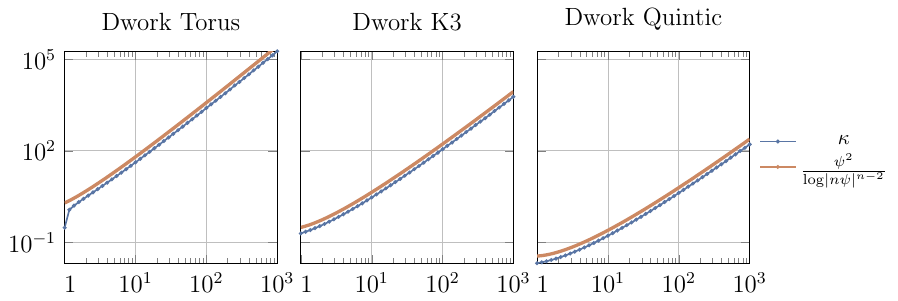}
    \caption{Numerical normalisation constant $\kappa$ matches predicted LCSL scaling of $\psi^2/\log\abs{n\psi}^{n-2}\times\text{const}$ for Dwork onefolds through threefolds. Analytic solution is slightly displaced to avoid overlapping numerical data.}
    \label{fig:kappa_scaling}
\end{figure}

Altogether, $K_\text{trop}$ arises either from Legendre-dual Monge--Amp\`ere \cref{eq:monge-ampere-real} or from direct computation using conjectured balanced potentials. It corresponds to the semi-flat metric, and can be used to derive power laws that match numerical computations. 

\subsection{Case Study: Dwork K3}

When switching to the 2-dimensional Dwork K3, we might hope to preserve many of the features we had observed on the torus. The most significant change occurs when attempting to solve the real Monge--Amp\`ere \cref{eq:monge-ampere-real}.

While not being directly solvable anymore, the equation is still quite rigid. By J\"orgens--Calabi--Pogorelov \cite[Section 4.3.1]{figalli_monge-ampere_2017}, the only solution that is both global and convex is once again the quadratic. Unfortunately, this is not applicable here. While we do know that $F$ is convex via convexity of $K_\text{trop}$, there is no reason for $F$ to extend to the entire $\RR^2$ since it is only defined on the 2-simplex $\Sigma_2=\braces*{\xi\in\RR_{\ge0}^3\colon \sum\xi_i=1}$.

There is also a geometric reason why $F$ cannot be a quadratic on K3. This would imply that its conjugate $K_\text{trop}$ is a quadratic, which implies not only $\det\Hess K_\text{trop}=\text{const}$, but a much stronger $\Hess K_\text{trop}=\text{const}$. Therefore, the resulting metric would be not only Ricci-flat but Riemann-flat, and such a metric cannot exist on $B\setminus\Delta$ since $B$ is isomorphic to the sphere.

Therefore, we should expect that the exponents follow a significantly more complicated pattern from the outset. Since there is no direct geometric way to constrain $F$, we turn to numerical data.

\begin{figure}[ht]
    \centering
    \includegraphics{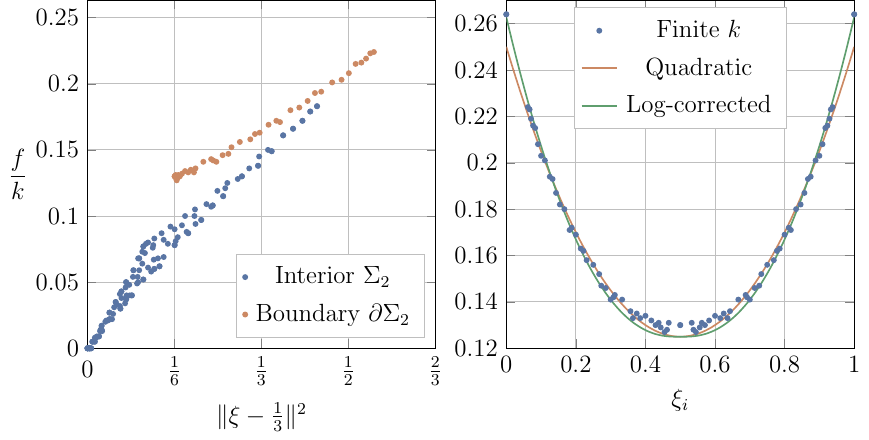}
    \caption{Left: Monomial exponents of degrees $1\le k\le16$ have a quadratic dependency on distance to center of $\Sigma_2$. Right: On a $\Sigma_1$ face of $\partial\Sigma_2$, the exponents behave like a quadratic with logarithmic corrections coming from the singular locus $\Delta$.}
    \label{fig:k3_allowed_sections}
\end{figure}

\cref{fig:k3_allowed_sections} shows that while $F$ does deviate from a quadratic on the entire $\Sigma_2$, it seems to do so only by a small amount. In fact, it appears exceptionally close to a quadratic on the boundary $\partial\Sigma_2$: 
\begin{equation}\label{eq:k3_riemann_baseline}
    F(\xi,1-\xi,0)\approx \frac{1}{8} + \frac{1}{2} \parens*{\xi-\frac{1}{2}}^2.
\end{equation}
This can be extended to the interior of $\Sigma$ with $F\parens*{\xi_1, \xi_2, \xi_3} = \parens*{\xi_1^2 + \xi_2^2 + \xi_3^2}/4$. However, this solution is Riemann-flat, not just Ricci-flat.

By Alexandrov comparison principle \cite[Section 2.3]{figalli_monge-ampere_2017}, specifying the boundary condition and exact determinant value in the interior forces the Monge--Amp\`ere \cref{eq:monge-ampere-real} to have an unique (weak) solution. Therefore, the only way to obtain a different behaviour in the interior while preserving the boundary is to change the determinant. 

However, changing the determinant is incompatible with the boundary and will create non-geometric creases in the interior of $\Sigma_2$. As a concrete example, the only arbitrary determinant solution that is $C^\infty$ at the boundary is
\begin{equation}
    F_t\parens*{\xi_1, \xi_2, \xi_3} = \parens*{\xi_1^2 + \xi_2^2 + \xi_3^2}/4 + t\parens*{\xi_1 \xi_2 \xi_3} / \max \xi_i.
\end{equation}
This solution not only has creases extending from the simplex vertices where $\max$ switches branches, but is also piecewise Riemann-flat. This cannot be a valid solution, and so the purely quadratic boundary condition with a smooth interior is ruled out. 

Recall that the singular locus $\Delta$ consists of 24 points \cite{gross_large_2000}. In the LCSL limit these points form six groups of four that are located on six edges of $\partial\Sigma_3$ \cite{kontsevich_affine_2006}. When LCSL is approached via the Dwork family, symmetry forces all 4 singularities to overlap on the midpoints. Restricting to one $\Sigma_2$ face, we conclude that $F$ has quadruple singularities at midpoints such as $(1/2,1/2,0)$.

Locally, these points are known to behave as logarithmic singularities \cite{ooguri_summing_1996}. In polar coordinates $(r,\theta)$ centred at each singularity, the base metric looks like $g_{rr} \sim -\log r$, and the potential that produces this metric is $K(r,\theta)\sim r^2 \log r$. For this reason, we attempt to expand $F$ as 
\begin{equation}\label{eq:k3_picard_fuchs_boundary}
    F(\xi,1-\xi,0)= c_0 + c_1 \parens*{\xi-\frac{1}{2}}^2 + c_2 \parens*{\xi-\frac{1}{2}}^2\log \abs*{\xi-\frac{1}{2}} + \mathcal{O}\parens*{\parens*{\xi-\frac{1}{2}}^4},
\end{equation}
where $c_0$, $c_1$, and $c_2$ are coefficients that we fit from data and show in \cref{tab:picard_fuchs_coeffs}.

\begin{table}[ht]
    \centering
    \begin{tabular}{c cc}
        \toprule
        Term & Numerical & Algebraic \\
        \midrule
        $1$ & $0.129{\scriptstyle\pm0.001}$ & $1/8$ \\
        $r^2$ & $0.631{\scriptstyle\pm0.017}$ & $2/3$ \\
        $r^2\log r$ & $0.170{\scriptstyle\pm0.020}$ & $1/6$ \\
        \bottomrule
    \end{tabular}
    \caption{Least Squares fit to \cref{eq:k3_picard_fuchs_boundary} on the boundary of $F$, with $r=\xi-1/2$ and $1\sigma$ errors.}
    \label{tab:picard_fuchs_coeffs}
\end{table}

We can also explain them using a topological heuristic. Since the singularities have multiplicity 4, each midpoint carries $4/24=1/6$ worth of curvature charge. Modifying the Riemann-flat baseline \cref{eq:k3_riemann_baseline} by $1/6$ yields $c_1=2/3$ and $c_2=1/6$. As \cref{tab:picard_fuchs_coeffs} shows, these topologically motivated fractions are very close to our numerical fits.


Having validated that the algebraic coefficients in \cref{tab:picard_fuchs_coeffs} are consistent with data, we can substitute them in \cref{eq:k3_picard_fuchs_boundary} to obtain the boundary condition $F$ must satisfy on each edge of $\partial\Sigma_2$:
\begin{equation}
    B(t) = \frac{1}{8} + \frac{2}{3}\parens*{t - \frac{1}{2}}^2 + \frac{1}{6}\parens*{t - \frac{1}{2}}^2 \log\abs*{t - \frac{1}{2}}.
\end{equation}
To find the full solution on $\Sigma_2$, we outline how to perform a Taylor expansion around the boundary. The hardest part is defining a baseline that has correct behaviour on the $\partial\Sigma_2$, while not introducing additional logarithmic singularities in the interior of $\Sigma_2$. We achieve this by treating the singular part $S$ and the remainder $R$ separately.

Let $P_1 = (0, \frac{1}{2}, \frac{1}{2})$, $P_2 = (\frac{1}{2}, 0, \frac{1}{2})$, and $P_3 = (\frac{1}{2}, \frac{1}{2}, 0)$ be the singular boundary midpoints. Furthermore, let $D_i(\xi) = \norm*{\xi - P_i}^2$ be the squared distances to them. Then, we can introduce 
\begin{equation}
    S(\xi_1, \xi_2, \xi_3) = \frac{1}{24} \sum_{i=1}^3 D_i \log(D_i)
\end{equation}
in order to capture the effect of singular points inside $\Sigma$. On each edge, $S$ captures its singularity, but it also contains the imprint of the other two. To get the remainder, we subtract $S$ from the boundary condition $B$ and introduce $R_{\partial\Sigma}(t) = B(t) - S(t, 1-t, 0)$. Then, this extends to the interior as
\begin{equation}
    R_\Sigma(\xi_1, \xi_2, \xi_3) = \frac{R_{\partial\Sigma}(\xi_1) + R_{\partial\Sigma}(\xi_2) + R_{\partial\Sigma}(\xi_3) - R_{\partial\Sigma}(0) }{2}.
\end{equation}
Our perturbative solution to $F$ then takes the form
\begin{equation}\label{eq:k3_boundary_expansion_F}
    F_\text{pert}(\xi_1, \xi_2, \xi_3) = S(\xi_1, \xi_2, \xi_3) + R_\Sigma(\xi_1, \xi_2, \xi_3) + \xi_1\xi_2\xi_3 \times \text{higher order terms.}
\end{equation}
\begin{prop}
$F_\mathrm{pert}$ satisfies the boundary condition $F^\mathrm{pert}(t, 1-t,0) = B(t)$.
\begin{proof}
Using $R_{\partial\Sigma}(t) = R_{\partial\Sigma}(1-t)$ which follows from symmetry, we calculate
\begin{align}
    F^\text{pert}(t,1-t,0) &= S(t, 1-t, 0) + R_\Sigma(t, 1-t, 0)\nonumber\\
    &= S(t, 1-t, 0) + \frac{R_{\partial\Sigma}(t) + R_{\partial\Sigma}(1-t) + R_{\partial\Sigma}(0) - R_{\partial\Sigma}(0) }{2} \nonumber\\
    &= S(t, 1-t, 0) + \parens*{B(t) - S(t, 1-t, 0)} \nonumber\\
    &= B(t),
        \end{align}
        as desired.
    \end{proof}
\end{prop}
At this point, one can parameterise higher order terms in \cref{eq:k3_boundary_expansion_F} with a symmetric Taylor series in $\xi$, and iteratively optimise for the coefficients to get as close to the true solution as desired. We opt for a simpler approach. 

We know that the dominant monomials are normalised to be constant, therefore their decay exponent is 0. Therefore, the barycenter $C=(\frac{1}{3}, \frac{1}{3}, \frac{1}{3})$ is a global minimum and satisfies $F(C) = 0$. Let
\begin{equation}\label{eq:k3_F_approx}
    F_\text{approx}(\xi_1, \xi_2, \xi_3) = S(\xi_1, \xi_2, \xi_3) + R_\Sigma(\xi_1, \xi_2, \xi_3) + \Lambda \xi_1\xi_2\xi_3,
\end{equation}
where
\begin{equation}
    \Lambda = -\frac{1}{27}\parens*{S(C) + R_\Sigma(C)}
\end{equation}
is chosen so that $F_\text{approx}(C)=0$. In this way, $F_\text{approx}$ is a first-order closed form approximation to the true $F$.

Significantly, $F_\text{approx}$ is a \emph{very good} approximation to the numerical data.

\begin{figure}[ht]
    \centering
    \includegraphics[width=\linewidth]{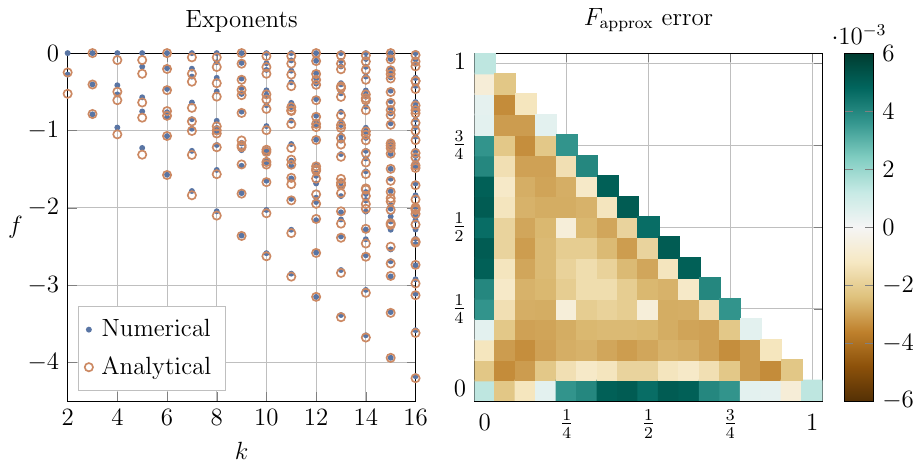}
    \caption{Left: First order approximation $F_\text{approx}$ matches numerical exponents. Right: Decay exponents of degree $k=12$ are within $5\times10^{-3}$ of asymptotic $F_\text{approx}$ prediction on $\Sigma$.}
    \label{fig:k3_lcsl_fit}
\end{figure}

\cref{fig:k3_lcsl_fit} shows that $kF_\text{approx}$ is a very good fit for low-degree exponents, with only $k=2$ being a significant outlier. The fit also remains very accurate for large $k$, even though both numerical data and $F_\text{approx}$ are only approximations to true $F$. In particular, for $k=12$, the data are within $5\times10^{-3}$ of each other on the simplex $\Sigma_2$, as \cref{fig:k3_lcsl_fit} (right) shows.

We end our exploration of Dwork K3 with one final example of mirror symmetry in action. So far, we have shown that the exponents contain information about $F(\xi)$, which is essentially the K\"ahler potential on the mirror manifold. Now, we will show how the two potentials numerically relate to each other.

\begin{figure}[ht]
    \centering
    \includegraphics[width=\linewidth]{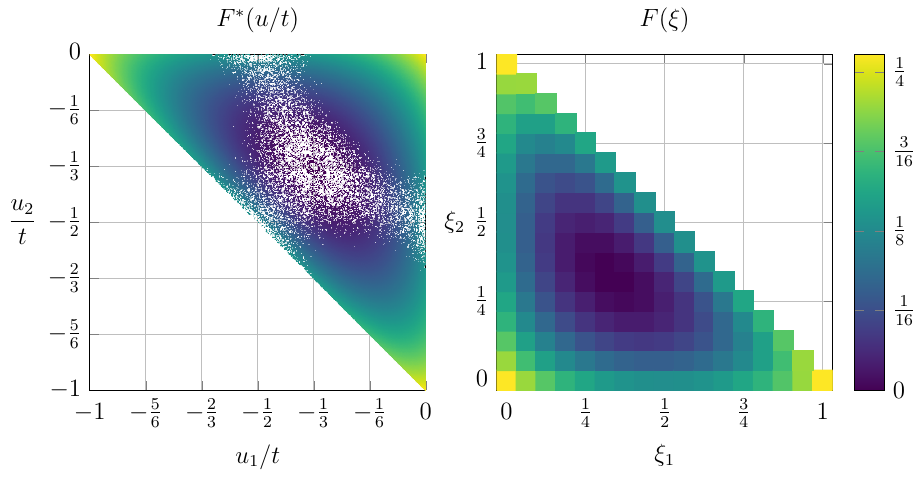}
    \caption{Legendre duality as a consequence of mirror symmetry. Left: $k=12$ balanced metric on $X_{10^3}$ projected onto a face of $B$. Right: Power laws from balanced exponents form a dual potential on the dual simplex. The colorbar is shared between the plots.}
    \label{fig:mirror_potentials}
\end{figure}

\cref{fig:mirror_potentials} (left) shows the continuous approximation to $F^* \sim K_\text{trop}$, obtained by computing the degree $k=12$ balanced metric on $X_{10^3}$. The transparent regions in the middle of the 2-simplex are sampling artefacts. On the other hand, \cref{fig:mirror_potentials} (right) computes $F$ based on the power-laws.

Our conjectural identification with mirror symmetry is therefore supported in multiple ways. Firstly, the coordinates on the left are $\log\abs{z_i} / \log \abs{n\psi}$ and are negative, while $\xi_i$ are positive. Hence, the two functions live on dual simplices. Secondly, while both plots appear visually similar, they are not identical but rather Legendre dual. The reason why they appear similar is because they are close to a pure quadratic, as we have already established, and quadratics are self-dual. In this way, the coefficients in balanced metrics contribute to continuous approximations of $K$ (left) and to discrete approximations of the dual (right).

\subsection{Case Study: Cefal\'u Quartics}

Next, we look at LCSL behaviour when approached using a different family. The Cefal\'u quartics \cite{catanese_kummer_2023,berglund_machine_2023} are 2-parameter K3 manifolds given by
\begin{equation}\label{eq:cefalu-def}
    Y_{\lambda,\psi}\colon \sum_{i=1}^4 z_i^4 - \frac{\lambda}{3} \bigg(\sum_{i=1}^4 z_i^2\bigg)^2 + 4\psi z_0 z_1 z_2 z_3 = 0.
\end{equation}
They generalise the Dwork family, since $Y_{0,\psi} = X_\psi$. The $(\lambda,\psi)$ locus where $Y$ fails to be transverse is fairly complex. We merely remark that $Y$ will have singularities for every $\psi\in\CC$ only when $\lambda\in\braces*{\frac{3}{2}, 3, \infty}$. This is sufficient for our purposes, since $\psi\to\infty$ remains the LCSL limit that we will be taking.

\begin{figure}[ht]
    \centering
    \includegraphics[width=\linewidth]{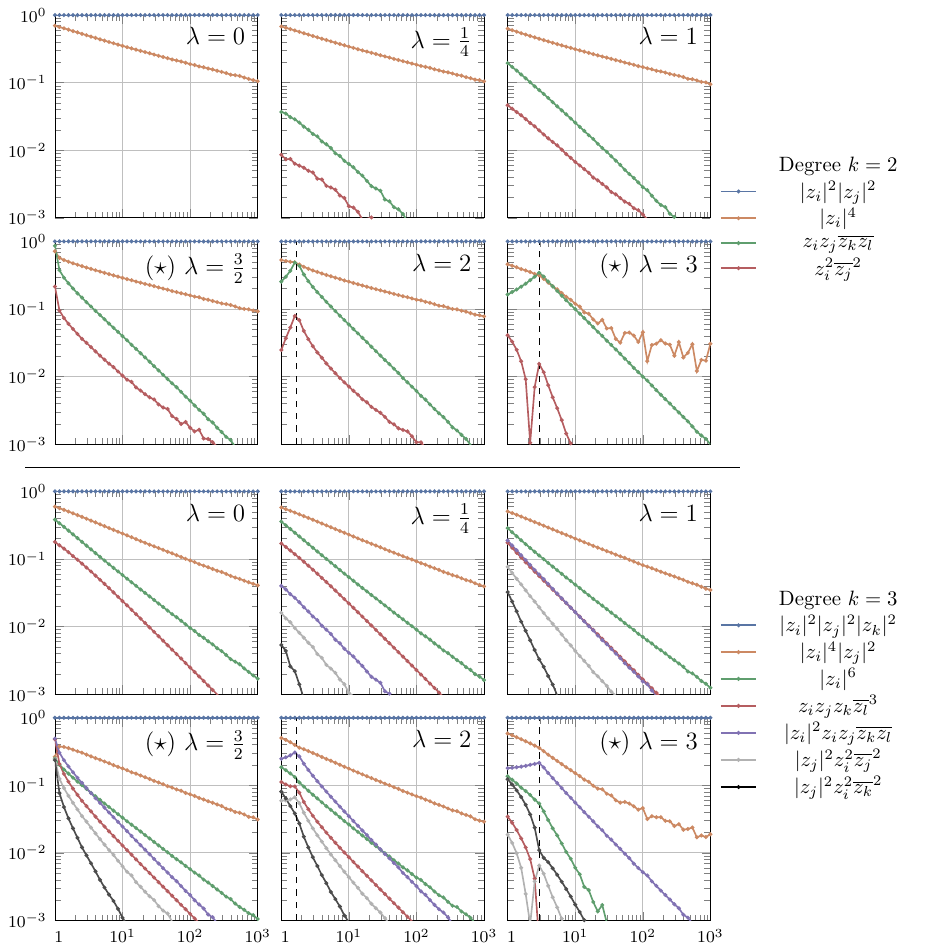}
    \caption{Power laws on the Cefal\`u family for different $\lambda$. Fully singular families are $\lambda=3/2,3$, marked with $(\star)$. Further singularities occur at $\psi=4\lambda/3\pm1$ and are marked with vertical dashed lines. }
    \label{fig:cefalu_lcsl}
\end{figure}

\cref{fig:cefalu_lcsl} shows the LCSL behaviour of this family for $k=2,3$, and has several notable features. Since \cref{eq:cefalu-def} has less symmetries, there are new monomials to keep track of. Furthermore, the singular manifolds at $\lambda=3$ are more noisy. 

Most importantly, though, we observe that allowed entries have the same behaviour across all choices of $\lambda$. This is consistent with our interpretation that they correspond to the Monge--Amp\`ere metric on $B$ and are thus the property of the LCSL limit itself.

Meanwhile, the forbidden entries have a strong $\lambda$ dependency, which is again expected since they encode the collapsing fibers that do depend on the direction LCSL is approached from.

\subsection{Case Study: CICY Threefold}

We conclude this section with a CICY given by the matrix
\begin{equation}
    \big[\, \PP^5 \,\big|\: 3~ 3 \,\big]_{-144}
\end{equation}
as our final example. Specifically, we study the one-parameter family
\begin{equation}
\begin{aligned}
    z_1^3 + z_2^3 + z_3^3 - 3\psi w_1w_2w_3 &= 0, \text{ and} \\
    w_1^3 + w_2^3 + w_3^3 - 3\psi z_1z_2z_3 &= 0,
\end{aligned}
\end{equation}
where $(z_1:z_2:z_3:w_1:w_2:w_3)$ are homogeneous coordinates in $\PP^5$. The reason we give different labels to them, even though they correspond to the same projective space, is because of how they interact with the two polynomials.

\begin{figure}[ht]
    \centering
    \includegraphics[]{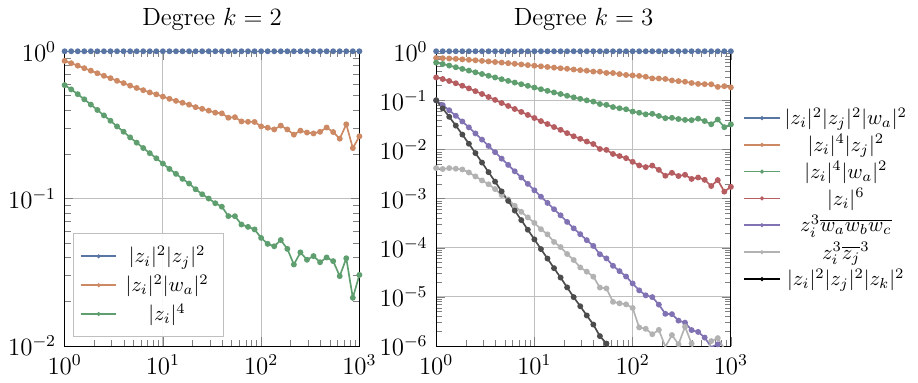}
    \caption{LCSL Power laws for the CICY threefold.}
    \label{fig:bicubic_lcsl}
\end{figure}

This family is invariant under permutations of $z$ and $w$ coordinates, as well as under their interchange. These form the wreath product $\SS_3 \wr \ZZ_2$. Like the Dwork family, toric phase symmetries $z_i\to \zeta_3^{k_i} z_i$, with the condition $\sum k_i = 0$, also exist. However, unlike the Dwork family, having two polynomials allows for existence of an order-9 invariant
\begin{equation}
    (z_1 : z_2 : z_3 : w_1 : w_2 : w_3) \mapsto (\zeta_3 z_1 : z_2 : z_3 : \zeta_9 w_1 : \zeta_9 w_2 : \zeta_9^{-2} w_3).
\end{equation}
In total, the automorphism group is $(\ZZ_9 \times \ZZ_3^2) \rtimes (\SS_3 \wr \ZZ_2)$ of order 5832.

\cref{fig:bicubic_lcsl} shows the behaviour of this family near the LCSL. The $k=2$ plot uses regular Donaldson's iteration, while $k=3$ requires our ambient method. We again observe that behaviour is consistent when switching to ambient iteration, and that power laws are a feature of CICYs as well.

Therefore, our computations and observations are successfully generalising to different families of Calabi--Yau manifolds.

%% file: 5_Discussion.tex
\section{Discussion}

In this paper we lifted Donaldson's algorithm to operate on global sections of the ambient space, as opposed to the variety. This approach preserves the canonical ambient basis of monomials of degree $k$, and therefore allows for more interpretable balanced metrics at large $k$. 

We predominantly studied balanced metrics on the Dwork family, and observed continuity in behaviour of these metrics when switching from the classical low-$k$ regime to the high $k$ where independent and full bases diverge.

We observe that near the Large Complex Structure Limit (LCSL) the entries of the balanced matrices $H$ obey novel power laws. We connect these power laws to the Gromov--Hausdorff metric collapse of the Ricci-flat metric onto the base of the special Lagrangian fibration.


Whereas Refs.~\cite{donaldson_numerical_2005,donaldson_scalar_2001} introduce balanced metrics, Ref.~\cite{douglas_numerical_2008} supplies further numerical and theoretical groundwork for computing them. Specifically, they identify the importance of quotienting the $H$ matrix with symmetries. However, they are ultimately working with $H^0(X,\mathcal{O}(k))$, not the more symmetric $H^0(\PP^n,\mathcal{O}(k))$. Consequently, on Dwork quintic with $n=5$, $k=12$, $N^X=1490$, and $N^\PP=1820$, they report restricting their $N^X\times N^X$ parameters down to 9800. In contrast, we reduce the larger $N^\PP\times N^\PP$ matrix to just 511.

Furthermore, when treating similar subjects, our approach yields stronger claims in a more direct manner. For example, while Ref.~\cite{douglas_numerical_2008} obtains that group invariants are preserved with proper initialisation, our \cref{prop:T+X-equivariant} proves group equivariance as a property of the entire $T^+_X$ map.

There are also many approaches that approximate Ricci-flat metrics using Donaldson's ansatz, but without running the algorithm itself. The benefits include being able to use the full monomial basis, implying easier handling of symmetries. Ref.~\cite{headrick_energy_2010} uses energy functionals to optimise Ricci-flatness directly and achieve significant reductions in $\sigma$-measure. 

Similarly, \texttt{CYJAX} of Ref.~\cite{gerdes_cyjax_2022} learns the ambient $H^\PP$ matrix with machine learning, while Ref.~\cite{ashmore_machine_2020} uses machine learning to compute higher $k$ from lower $k$. By contrast, our approach provides a novel way to use the full monomial basis within Donaldson's algorithm itself.

Many other numerical Calabi--Yau approaches focus only on approximating Ricci-flat metrics, and so trade the algebraic ansatz for a neural network parameterisation of the metric \cite{Anderson:2020hux,larfors_learning_2021,berglund_cymyc_2024,ashmore_machine_2020,douglas_numerical_2021,hendi_learning_2024,jejjala_neural_2022,berglund_machine_2023}. This has in turn inspired interpretability work for these models \cite{mirjanic_symbolic_2025}.

The interpretability for balanced metrics remains scarce. Ref.~\cite{ek_calabi-yau_2024} employs a Grassmannian based approach with some interpretability consequences. Ref.~\cite{constantin_calabi-yau_2026} uses a similar ansatz to study K\"ahler moduli dependence and find symbolic approximations to the parameters. Ref.~\cite{lee_approximate_2025} optimises the algebraic ansatz with gradient descent and fits an exponential relationship to the complex structure moduli, and does not analyse the behaviour at the LCSL. By contrast, the behaviour we predict for the balanced metrics is exact, not approximate. Our power laws exist at every level $k$, and they form an experimental validation of the Gromov--Hausdorff metric collapse. 

 Thanks to Refs.~\cite{gross_large_2000, gross_real_2011, kontsevich_homological_2001, li_intermediate_2025, li_metric_2024, li_stromingeryauzaslow_2022, gross_collapsing_2013, boucksom_tropical_2017} and others, many properties of the metric collapse and LCSL behaviour have already been established, such as boundedness of curvature away from discriminant locus, or non-Archimedean analogues. Our power laws provide a new way to numerically compute some properties of these theories.

Finally, connections between mirror symmetry and Legendre duality trace back to Ref.~\cite{hitchin_moduli_1997}, and specifically in the current context, Gross and Siebert showed that a polarised toric degeneration gives a pair of dual affine manifolds related by a discrete Legendre transform~\cite{gross_siebert_affine_2003,gross_siebert_mirror_I_2006}. Legendre transforms and Ronkin functions are tightly coupled to the geometry of amoebas~\cite{gelfand_discriminants_1994} and to counting functions of gauge instantons~\cite{maeda_amoebas_2007}. Our connection to exponents of balanced metrics is a tangible manifestation of this deep theoretic framework, and a way to move from abstract theorems to computationally tractable phenomena. 
In particular, the monomial sections of ${\cal O}(k)$ on $\PP^n$ that form our ambient basis are precisely the leading-order ($q\to 0$) instances of the theta functions constructed by Gross and Siebert \cite{gross_siebert_theta_2016}, providing a theoretical motivation for why this basis is canonical. 




\acknowledgments
We are grateful to participants at ``The Unreasonable Effectiveness of Toric Geometry'' program at the Erwin Schr\"odinger International Institute for Mathematics and Physics, Vienna, and to participants at the ``AI $\times$ Mathematics 2026'' workshop at ICMS, Edinburgh, for helpful discussions about this work.

PB thanks the CERN Theoretical Physics Department for support over the years.
TH is grateful to the Mathematics Department of the University of Maryland and to the Physics Department of the University of Novi Sad, Serbia, for recurring hospitality and resources.
VJ is supported by the South African Research Chairs Initiative of the Department of Science, Technology, and Innovation and the National Research Foundation (grant 78554). VM is supported by a studentship from Trinity College, University of Cambridge.
CM is supported by a Turing AI grant ($713$ $G111021$).

%% file: Appendix.tex
\section{Implementation Details}\label{app:implementation}

For our experiments we sample $N=10^6$ points with respect to reference measure. For hypersurfaces in $\PP^n$ this reference is just the pulled-back Fubini--Study metric, while for CICYs it is more complicated since point sampling is implemented \cite{larfors_learning_2021,berglund_cymyc_2024} via a theorem of \cite{shiffman_distribution_1999}.

As is usual, we are compensating for the sampling bias by computing the integration weights $w$. However, in the $\psi\to\infty$ limit, the reference measures become degenerate. We measure this collapse by computing the Kish's \emph{effective sample size}
\begin{equation}
    N_\text{eff} = \frac{\parens*{\sum w_i}^2}{\sum w_i^2},
\end{equation}
and observe that $N_\text{eff}$ is reasonably well behaved up to $\psi=10^3$, the largest $\psi$ in our experiments. For the quintic, K3 and the torus, we measure $N_\text{eff}/N\approx0.1,0.07,0.05$, respectively at this far end. The complex torus is the only manifold of the three that experiences a catastrophic drop in $N_\text{eff}$ for any choice of $\psi$. This happens at the conifold point $\psi=1$, at which the weights 
\begin{equation}
    w = \frac{\parens*{\abs{z_1}^2 + \abs{z_2}^2 + \abs{z_3}^2}^2}{\abs{z_1^2 - z_2z_3}^2 + \abs{z_2^2 - z_3z_1}^2 + \abs{z_3^2 - z_1z_1}^2}
\end{equation}
localise to the points where $z_i$ are third roots of unity because that is where the denominator diverges. The above formula is adapted from Ref.~\cite{mirjanic_symbolic_2025}.

When taking the pseudo-inverse with \verb|jax.numpy.linalg.pinv|, we use the default singular value threshold since the singular eigenvalues are not affected by Monte-Carlo error. We also pass \verb|hermitian=True| since that is guaranteed by the algorithm. At the end, we assert that the computed $H^\PP$ matrix is positive semi-definite by checking $\abs{\Im \lambda_i}\le10^{-8}$ and $\min \Re \lambda_i > -10^{-8}$. 

For small $k$, pseudo-inverse reduces to regular inverse so we compute exact values as the regular Donaldson's algorithm, which we also confirmed numerically. For large $k$, we performed multiple checks that our fixed point $H^\PP$ equals $(R^+)^\dag H^X R^+$, for some basis $s^X$. However, we prefer our $T^+_X$-map since it avoids having to compute $R$.

We terminate our iterations either when relative change in entries of $H^\PP$ drops below $\delta:=1/\sqrt{N}=10^{-3}$, or after a set number of steps, and observe that our outputs have very little variance under resampling even for high $k$. The threshold $\delta$ is chosen to match the Monte Carlo integration error, which is $\mathcal{O}\parens[\big]{1/\sqrt{N}}$.

To estimate the final uncertainty in $H_{ij}$ we track the standard deviation of averaging $G_{ij}$ entries that are known to be identical due to symmetries. We observe that this tends to be not much greater than $\delta$, although it does get worse for large $k$.

Finally, we compute the errors of linear fits for the power laws near LCSL and observe that they are very small. Variance in these linear fits contains the upstream noise from Donaldson's iteration, so it is another signal that our algorithm converges properly.

To plot the power laws near LCSL in \cref{fig:bicubic_lcsl,fig:cefalu_lcsl,fig:donaldson_coefficients_inf}, we sample 46 values of $\psi\in\bracks{1,10^3}$ so that $\log\psi$ is uniformly distributed. However, we use different data to compute power law exponents in \cref{tab:k3_lcsl_appendix,tab:quintic_lcsl_appendix,tab:torus_lcsl_appendix,fig:k3_lcsl_fit,fig:k3_allowed_sections,fig:torus_allowed_sections}, and in all downstream calculations. 

Firstly, we build our dataset by sampling 100 points with $\psi\in\bracks*{1,10^3}$. Then, in order to avoid the distortions that come from both the conifold and the LCSL, we use the middle third (34 points in range $\bracks*{10^1,10^2}$) for linear regression $\log \abs{H_{ij}} \sim a + b \log\abs{n\psi}$ to extract the relationships. 
We collect $a$ and $b$ values for Dwork onefolds up to $k=18$, twofolds up to $k=12$ and threefolds up to $k=9$, and present them in \cref{app:raw_fit_data}.

We find that restricting $\psi$ to $\bracks{3,100}$ is a good enough heuristic that avoids both the conifold point and numerical issues for too big $\psi$. We experimented with more complicated pipelines that used Gaussian Processes to estimate the heteroscedastic noise on a larger range, but the extra complexity did not bring any meaningful difference. 

Finally, we perform various optimisations to improve our runtime. Since $H$ is continuous in complex structure moduli, a fixed point for some $\psi$ is expected to be an excellent guess for $\psi+\epsilon$. Therefore, we only iterate from the Fubini--Study metric for the first manifold, and reuse previous fixed points for subsequent ones.

Secondly, we observe that the $G$ and $H$ matrices will always be block-diagonal in our use-case, due to presence of toric $\ZZ_n$ isometries. Recall that on the Dwork family $H\bracks{z^\alpha \overline{z^\beta}}\ne0$ only if $\alpha_i-\beta_i=\text{const}\pmod n$. This condition is transitive, so if we also have $H\bracks{z^\beta \overline{z^\gamma}}\ne0$ then $H\bracks{z^\alpha \overline{z^\gamma}}\ne0$ too. This is therefore an equivalence relation that clusters sections $z^\alpha$ into disjoint classes. More formally, if the character $\chi_\alpha(g)$ is the phase $z^\alpha$ picks up when transformed by the toric isometry $g$ then the above equivalence relation is simply given by equality of characters.

Dwork family has $\abs{\ZZ_n^{n-2}}=n^{n-2}$ characters in total, so there are up to that many blocks in $G$ and $H$. If $k<n-1$, the mapping $z^\alpha\mapsto\chi_\alpha$ is injective so $H$ is purely diagonal with every section belonging to a separate class. For $k\ge n-1$, the maximum number of characters is reached. For example, non-trivial classes at $k=n-1$ are $\braces{z_i^{n-1}, \prod_{j\ne i} z_j}$. Specifically for the quintic, the computed $H$ matrix will have up to 125 blocks on the diagonal. Since computing the outer product $s^\PP (s^\PP)^\dag$ in the integral is both time and memory intensive, exploiting the block-diagonal structure has significant impact.
 
Finally, we are just-in-time compiling the entire loop with \verb|jax.jit|. This lowers our python code into highly optimised machine code, without us having to write it ourselves.

All of these optimisations along with hardware improvements over the years amount to us computing $k=12$ balanced metrics on the Dwork quintic in under 5 min per manifold. By contrast, Ref.~\cite{douglas_numerical_2008} reported that computing a single $k=12$ metric took 2 days in 2008. Admittedly, we do need 29.1GB RAM in order to keep sections for 1M points precomputed in memory, while their implementation only used 4GB.

\section{Power Laws at LCSL}\label{app:raw_fit_data}

The following tables contain linear regression fits to low-degree sections. The reported errors are $1\sigma$ errors of linear fits and do not include errors from the rest of the pipeline such as convergence to flat metrics or Monte Carlo integration errors.

\begin{table}[htp]
    \centering
    \input{best_power_laws_torus}
    \caption{Best fits for parameters in the balanced metrics on the torus.}
    \label{tab:torus_lcsl_appendix}
\end{table}

\begin{table}[htp]
    \centering
    \input{best_power_laws_k3}
    \caption{Best fits for parameters in the balanced metrics on the K3.}
    \label{tab:k3_lcsl_appendix}
\end{table}

\begin{table}[htp]
    \centering
    \input{best_power_laws_quintic}
    \caption{Best fits for parameters in the balanced metrics on the Quintic.}
    \label{tab:quintic_lcsl_appendix}
\end{table}

%% file: best_power_laws_torus.tex

\begin{tabular}{cc cc}
\toprule
\multicolumn{2}{c}{
Torus Sections} & \multicolumn{2}{c}{$H^{\PP}_{ij} \sim \alpha \parens*{n\abs{\psi}}^{\beta}$} \\
\cmidrule(rl){3-4}
{$k$} & {$s^\PP \otimes \overline{s^\PP}$} & {$\alpha$} & {$-\beta$} \\
\cmidrule(r){1-2} \cmidrule(rl){3-4}
\multirow{3}*{\small{$2$}}&\footnotesize{$\abs{z_i}^{2}\abs{z_j}^{2}$}&\small{1}&\small{0}\\
&\footnotesize{$\abs{z_i}^{4}$}&$1.025{\scriptstyle \pm0.005}$&$0.501{\scriptstyle \pm0.001}$\\
&\footnotesize{$z_i^{2}\overline{z_j}\overline{z_k}$}&$1.018{\scriptstyle \pm0.019}$&$1.003{\scriptstyle \pm0.004}$\\
\cmidrule(r){2-2} \cmidrule(rl){3-4}
\multirow{6}*{\small{$3$}}&\footnotesize{$\abs{z_i}^{4}\abs{z_j}^{2}$}&\small{1}&\small{0}\\
&\footnotesize{$\abs{z_i}^{6}$}&$0.972{\scriptstyle \pm0.007}$&$0.658{\scriptstyle \pm0.002}$\\
&\footnotesize{$\abs{z_j}^{2}z_i^{2}\overline{z_j}\overline{z_k}$}&$1.021{\scriptstyle \pm0.033}$&$1.004{\scriptstyle \pm0.007}$\\
&\footnotesize{$\abs{z_i}^{2}z_i^{2}\overline{z_j}\overline{z_k}$}&$0.933{\scriptstyle \pm0.012}$&$1.650{\scriptstyle \pm0.003}$\\
&\footnotesize{$\abs{z_i}^{2}\abs{z_j}^{2}\abs{z_k}^{2}$}&$2.799{\scriptstyle \pm0.012}$&$2.650{\scriptstyle \pm0.003}$\\
&\footnotesize{$z_i^{3}\overline{z_j}^{3}$}&$2.822{\scriptstyle \pm0.262}$&$2.773{\scriptstyle \pm0.079}$\\
\cmidrule(r){2-2} \cmidrule(rl){3-4}
\multirow{12}*{\small{$4$}}&\footnotesize{$\abs{z_i}^{4}\abs{z_j}^{4}$}&\small{1}&\small{0}\\
&\footnotesize{$\abs{z_i}^{6}\abs{z_j}^{2}$}&$0.969{\scriptstyle \pm0.007}$&$0.241{\scriptstyle \pm0.002}$\\
&\footnotesize{$\abs{z_i}^{8}$}&$0.966{\scriptstyle \pm0.011}$&$0.988{\scriptstyle \pm0.003}$\\
&\footnotesize{$\abs{z_i}^{2}\abs{z_j}^{2}z_i^{2}\overline{z_j}\overline{z_k}$}&$0.939{\scriptstyle \pm0.010}$&$1.235{\scriptstyle \pm0.003}$\\
&\footnotesize{$\abs{z_i}^{4}z_j^{2}\overline{z_i}\overline{z_k}$}&$1.138{\scriptstyle \pm0.118}$&$1.273{\scriptstyle \pm0.026}$\\
&\footnotesize{$\abs{z_i}^{4}z_i^{2}\overline{z_j}\overline{z_k}$}&$0.820{\scriptstyle \pm0.023}$&$1.955{\scriptstyle \pm0.006}$\\
&\footnotesize{$z_i^{4}\overline{z_j}^{2}\overline{z_k}^{2}$}&$0.883{\scriptstyle \pm0.030}$&$1.974{\scriptstyle \pm0.006}$\\
&\footnotesize{$\abs{z_j}^{2}z_i^{3}\overline{z_k}^{3}$}&$1.893{\scriptstyle \pm0.038}$&$2.236{\scriptstyle \pm0.008}$\\
&\footnotesize{$\abs{z_i}^{4}\abs{z_j}^{2}\abs{z_k}^{2}$}&$2.100{\scriptstyle \pm0.007}$&$2.255{\scriptstyle \pm0.002}$\\
&\footnotesize{$\abs{z_j}^{2}\abs{z_k}^{2}z_i^{2}\overline{z_j}\overline{z_k}$}&$2.329{\scriptstyle \pm0.146}$&$2.277{\scriptstyle \pm0.032}$\\
&\footnotesize{$\abs{z_i}^{2}z_i^{3}\overline{z_j}^{3}$}&$1.460{\scriptstyle \pm0.234}$&$2.339{\scriptstyle \pm0.070}$\\
&\footnotesize{$\abs{z_j}^{2}z_i^{3}\overline{z_j}^{3}$}&$1.460{\scriptstyle \pm0.234}$&$2.339{\scriptstyle \pm0.070}$\\
\bottomrule
\end{tabular}

%% file: best_power_laws_k3.tex

\begin{tabular}{cc cc}
\toprule
\multicolumn{2}{c}{
K3 Sections} & \multicolumn{2}{c}{$H^{\PP}_{ij} \sim \alpha \parens*{n\abs{\psi}}^{\beta}$} \\
\cmidrule(rl){3-4}
{$k$} & {$s^\PP \otimes \overline{s^\PP}$} & {$\alpha$} & {$-\beta$} \\
\cmidrule(r){1-2} \cmidrule(rl){3-4}
\multirow{2}*{\small{$2$}}&\footnotesize{$\abs{z_i}^{2}\abs{z_j}^{2}$}&\small{1}&\small{0}\\
&\footnotesize{$\abs{z_i}^{4}$}&$0.982{\scriptstyle \pm0.006}$&$0.276{\scriptstyle \pm0.001}$\\
\cmidrule(r){2-2} \cmidrule(rl){3-4}
\multirow{4}*{\small{$3$}}&\footnotesize{$\abs{z_i}^{2}\abs{z_j}^{2}\abs{z_k}^{2}$}&\small{1}&\small{0}\\
&\footnotesize{$\abs{z_i}^{4}\abs{z_j}^{2}$}&$1.050{\scriptstyle \pm0.005}$&$0.401{\scriptstyle \pm0.001}$\\
&\footnotesize{$\abs{z_i}^{6}$}&$1.069{\scriptstyle \pm0.011}$&$0.788{\scriptstyle \pm0.003}$\\
&\footnotesize{$z_i^{3}\overline{z_j}\overline{z_k}\overline{z_l}$}&$0.844{\scriptstyle \pm0.010}$&$0.970{\scriptstyle \pm0.002}$\\
\cmidrule(r){2-2} \cmidrule(rl){3-4}
\multirow{9}*{\small{$4$}}&\footnotesize{$\abs{z_i}^{4}\abs{z_j}^{2}\abs{z_k}^{2}$}&\small{1}&\small{0}\\
&\footnotesize{$\abs{z_i}^{4}\abs{z_j}^{4}$}&$1.069{\scriptstyle \pm0.003}$&$0.413{\scriptstyle \pm0.001}$\\
&\footnotesize{$\abs{z_i}^{6}\abs{z_j}^{2}$}&$1.061{\scriptstyle \pm0.003}$&$0.532{\scriptstyle \pm0.001}$\\
&\footnotesize{$\abs{z_i}^{8}$}&$1.005{\scriptstyle \pm0.003}$&$0.961{\scriptstyle \pm0.001}$\\
&\footnotesize{$\abs{z_j}^{2}z_i^{3}\overline{z_j}\overline{z_k}\overline{z_l}$}&$0.826{\scriptstyle \pm0.009}$&$0.970{\scriptstyle \pm0.002}$\\
&\footnotesize{$z_i^{4}\overline{z_j}^{4}$}&$0.195{\scriptstyle \pm0.043}$&$1.744{\scriptstyle \pm0.012}$\\
&\footnotesize{$\abs{z_i}^{2}z_i^{3}\overline{z_j}\overline{z_k}\overline{z_l}$}&$1.119{\scriptstyle \pm0.004}$&$1.979{\scriptstyle \pm0.001}$\\
&\footnotesize{$z_i^{2}z_j^{2}\overline{z_k}^{2}\overline{z_l}^{2}$}&$2.301{\scriptstyle \pm0.394}$&$2.446{\scriptstyle \pm0.089}$\\
&\footnotesize{$\abs{z_i}^{2}\abs{z_j}^{2}\abs{z_k}^{2}\abs{z_l}^{2}$}&$4.476{\scriptstyle \pm0.004}$&$2.979{\scriptstyle \pm0.001}$\\
\cmidrule(r){2-2} \cmidrule(rl){3-4}
\multirow{15}*{\small{$5$}}&\footnotesize{$\abs{z_i}^{4}\abs{z_j}^{4}\abs{z_k}^{2}$}&\small{1}&\small{0}\\
&\footnotesize{$\abs{z_i}^{6}\abs{z_j}^{2}\abs{z_k}^{2}$}&$1.006{\scriptstyle \pm0.002}$&$0.174{\scriptstyle \pm0.000}$\\
&\footnotesize{$\abs{z_i}^{6}\abs{z_j}^{4}$}&$1.063{\scriptstyle \pm0.003}$&$0.569{\scriptstyle \pm0.001}$\\
&\footnotesize{$\abs{z_i}^{8}\abs{z_j}^{2}$}&$1.003{\scriptstyle \pm0.003}$&$0.752{\scriptstyle \pm0.001}$\\
&\footnotesize{$\abs{z_j}^{2}\abs{z_k}^{2}z_i^{3}\overline{z_j}\overline{z_k}\overline{z_l}$}&$0.930{\scriptstyle \pm0.008}$&$0.986{\scriptstyle \pm0.002}$\\
&\footnotesize{$\abs{z_j}^{4}z_i^{3}\overline{z_j}\overline{z_k}\overline{z_l}$}&$0.707{\scriptstyle \pm0.012}$&$1.125{\scriptstyle \pm0.003}$\\
&\footnotesize{$\abs{z_i}^{10}$}&$0.985{\scriptstyle \pm0.004}$&$1.224{\scriptstyle \pm0.001}$\\
&\footnotesize{$\abs{z_j}^{2}z_i^{4}\overline{z_j}^{4}$}&$0.253{\scriptstyle \pm0.039}$&$1.666{\scriptstyle \pm0.011}$\\
&\footnotesize{$\abs{z_i}^{2}z_i^{4}\overline{z_j}^{4}$}&$0.253{\scriptstyle \pm0.039}$&$1.666{\scriptstyle \pm0.011}$\\
&\footnotesize{$\abs{z_i}^{2}\abs{z_j}^{2}z_i^{3}\overline{z_j}\overline{z_k}\overline{z_l}$}&$1.032{\scriptstyle \pm0.002}$&$1.757{\scriptstyle \pm0.001}$\\
&\footnotesize{$\abs{z_k}^{2}z_i^{2}z_j^{2}\overline{z_k}^{2}\overline{z_l}^{2}$}&$0.777{\scriptstyle \pm0.213}$&$1.985{\scriptstyle \pm0.044}$\\
&\footnotesize{$\abs{z_i}^{2}z_i^{2}z_j^{2}\overline{z_k}^{2}\overline{z_l}^{2}$}&$0.777{\scriptstyle \pm0.213}$&$1.985{\scriptstyle \pm0.044}$\\
&\footnotesize{$\abs{z_j}^{2}z_i^{4}\overline{z_k}^{4}$}&$0.044{\scriptstyle \pm0.752}$&$2.153{\scriptstyle \pm0.190}$\\
&\footnotesize{$\abs{z_i}^{4}z_i^{3}\overline{z_j}\overline{z_k}\overline{z_l}$}&$1.365{\scriptstyle \pm0.007}$&$2.272{\scriptstyle \pm0.002}$\\
&\footnotesize{$\abs{z_i}^{4}\abs{z_j}^{2}\abs{z_k}^{2}\abs{z_l}^{2}$}&$3.655{\scriptstyle \pm0.006}$&$2.783{\scriptstyle \pm0.001}$\\
\bottomrule
\end{tabular}

%% file: best_power_laws_quintic.tex

\begin{tabular}{cc cc}
\toprule
\multicolumn{2}{c}{
Quintic Sections} & \multicolumn{2}{c}{$H^{\PP}_{ij} \sim \alpha \parens*{n\abs{\psi}}^{\beta}$} \\
\cmidrule(rl){3-4}
{$k$} & {$s^\PP \otimes \overline{s^\PP}$} & {$\alpha$} & {$-\beta$} \\
\cmidrule(r){1-2} \cmidrule(rl){3-4}
\multirow{2}*{\small{$2$}}&\footnotesize{$\abs{z_i}^{2}\abs{z_j}^{2}$}&\small{1}&\small{0}\\
&\footnotesize{$\abs{z_i}^{4}$}&$0.979{\scriptstyle \pm0.002}$&$0.192{\scriptstyle \pm0.000}$\\
\cmidrule(r){2-2} \cmidrule(rl){3-4}
\multirow{3}*{\small{$3$}}&\footnotesize{$\abs{z_i}^{2}\abs{z_j}^{2}\abs{z_k}^{2}$}&\small{1}&\small{0}\\
&\footnotesize{$\abs{z_i}^{4}\abs{z_j}^{2}$}&$1.033{\scriptstyle \pm0.004}$&$0.250{\scriptstyle \pm0.001}$\\
&\footnotesize{$\abs{z_i}^{6}$}&$0.969{\scriptstyle \pm0.006}$&$0.504{\scriptstyle \pm0.001}$\\
\cmidrule(r){2-2} \cmidrule(rl){3-4}
\multirow{6}*{\small{$4$}}&\footnotesize{$\abs{z_i}^{2}\abs{z_j}^{2}\abs{z_k}^{2}\abs{z_l}^{2}$}&\small{1}&\small{0}\\
&\footnotesize{$\abs{z_i}^{4}\abs{z_j}^{2}\abs{z_k}^{2}$}&$1.148{\scriptstyle \pm0.002}$&$0.352{\scriptstyle \pm0.001}$\\
&\footnotesize{$\abs{z_i}^{4}\abs{z_j}^{4}$}&$1.164{\scriptstyle \pm0.004}$&$0.606{\scriptstyle \pm0.001}$\\
&\footnotesize{$\abs{z_i}^{6}\abs{z_j}^{2}$}&$1.119{\scriptstyle \pm0.004}$&$0.676{\scriptstyle \pm0.001}$\\
&\footnotesize{$z_i^{4}\overline{z_j}\overline{z_k}\overline{z_l}\overline{z_m}$}&$0.625{\scriptstyle \pm0.023}$&$0.923{\scriptstyle \pm0.005}$\\
&\footnotesize{$\abs{z_i}^{8}$}&$1.089{\scriptstyle \pm0.005}$&$0.975{\scriptstyle \pm0.001}$\\
\cmidrule(r){2-2} \cmidrule(rl){3-4}
\multirow{10}*{\small{$5$}}&\footnotesize{$\abs{z_i}^{4}\abs{z_j}^{2}\abs{z_k}^{2}\abs{z_l}^{2}$}&\small{1}&\small{0}\\
&\footnotesize{$\abs{z_i}^{4}\abs{z_j}^{4}\abs{z_k}^{2}$}&$1.136{\scriptstyle \pm0.002}$&$0.355{\scriptstyle \pm0.001}$\\
&\footnotesize{$\abs{z_i}^{6}\abs{z_j}^{2}\abs{z_k}^{2}$}&$1.095{\scriptstyle \pm0.002}$&$0.447{\scriptstyle \pm0.000}$\\
&\footnotesize{$\abs{z_i}^{6}\abs{z_j}^{4}$}&$1.088{\scriptstyle \pm0.004}$&$0.694{\scriptstyle \pm0.001}$\\
&\footnotesize{$\abs{z_i}^{8}\abs{z_j}^{2}$}&$1.067{\scriptstyle \pm0.003}$&$0.814{\scriptstyle \pm0.001}$\\
&\footnotesize{$\abs{z_j}^{2}z_i^{4}\overline{z_j}\overline{z_k}\overline{z_l}\overline{z_m}$}&$0.621{\scriptstyle \pm0.023}$&$0.926{\scriptstyle \pm0.005}$\\
&\footnotesize{$\abs{z_i}^{10}$}&$0.984{\scriptstyle \pm0.005}$&$1.131{\scriptstyle \pm0.001}$\\
&\footnotesize{$z_i^{5}\overline{z_j}^{5}$}&$0.529{\scriptstyle \pm0.020}$&$1.858{\scriptstyle \pm0.005}$\\
&\footnotesize{$\abs{z_i}^{2}z_i^{4}\overline{z_j}\overline{z_k}\overline{z_l}\overline{z_m}$}&$1.452{\scriptstyle \pm0.016}$&$2.195{\scriptstyle \pm0.004}$\\
&\footnotesize{$\abs{z_i}^{2}\abs{z_j}^{2}\abs{z_k}^{2}\abs{z_l}^{2}\abs{z_m}^{2}$}&$7.259{\scriptstyle \pm0.016}$&$3.195{\scriptstyle \pm0.004}$\\
\bottomrule
\end{tabular}